\documentclass[a4paper,twocolumn,superscriptaddress,11pt,accepted=2019-01-23]{quantumarticle}
\pdfoutput=1

\usepackage{graphicx,color}
\usepackage{amsmath,amsfonts,enumerate,amsthm,amssymb,bbm}
\usepackage{hyperref}
\usepackage{multirow}
\usepackage{bbold}
\usepackage{braket}
\usepackage[utf8]{inputenc}
\usepackage[T1]{fontenc}
\usepackage[english]{babel}

\usepackage[numbers, sort&compress]{natbib}

\hypersetup{
   colorlinks=true,         
   linkcolor=blue,          
   citecolor=blue,          
}

\def\>{\rangle}
\def\<{\langle}
\def\K{ {\mathcal K} }
\def\E{ {\mathcal E} }
\def\P{ {\mathcal P} }
\def\H{ {\mathcal H} }
\def\M{ {\mathcal M} }

\def\U {{\mathcal U}}

\def\F{ {\mathcal F} }
\def\B{ {\mathcal B} }

\def\P{ {\mathcal P} }
\def\D{ {\mathcal D} }

\def\I{ \mathbbm{1} }
\def\tr{ \mbox{tr} }

\def\x{\boldsymbol{x}}

\def\z{\boldsymbol{z}}

\usepackage{dsfont}

\usepackage{physics}

\newtheorem{theorem}{Theorem}

\newtheorem{lemma}{Lemma}

\usepackage{color}

\begin{document}

\title{On the classification of two-qubit group orbits and the use of coarse-grained `shape' as a superselection property}

\author{Thomas Hebdige}
\affiliation{Controlled Quantum Dynamics Theory Group, Imperial College London, Prince Consort Road, London SW7 2BW, UK}
\author{David Jennings}
\affiliation{Controlled Quantum Dynamics Theory Group, Imperial College London, Prince Consort Road, London SW7 2BW, UK}
\affiliation{Department of Physics, University of Oxford, Oxford, OX1 3PU, UK}
\affiliation{School of Physics and Astronomy, University of Leeds, Leeds, LS2 9JT, UK.}
\date{\today}

\begin{abstract}
Recently a complete set of entropic conditions has been derived for the interconversion structure of states under quantum operations that respect a specified symmetry action, however the core structure of these conditions is still only partially understood. Here we develop a coarse-grained description with the aim of shedding light on both the structure and the complexity of this general problem. Specifically, we consider the degree to which one can associate a basic `shape' property to a quantum state or channel that captures coarse-grained data either for state interconversion or for the use of a state within a simulation protocol. We provide a complete solution for the two-qubit case under the rotation group, give analysis for the more general case and discuss possible extensions of the approach. 
\end{abstract}

\maketitle

\begin{figure}[t]
\includegraphics[width=8.5cm]{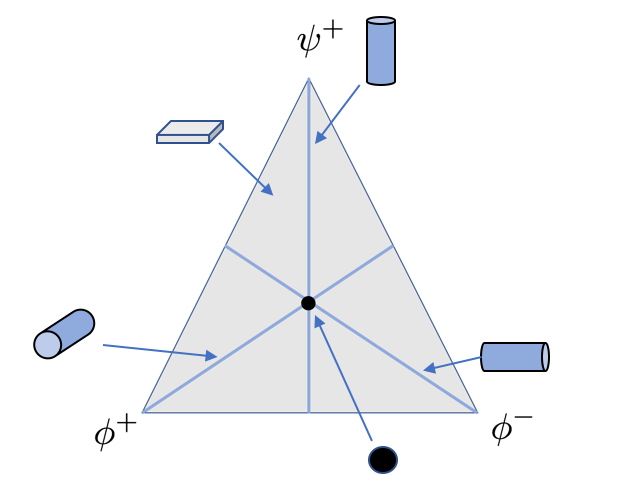}
\caption{\textbf{Residual symmetry-types in two-qubit systems.} We associate to each state a `shape' that describes the residual symmetries of that state. As an example, in the above figure we show the symmetry properties of the set of all mixtures of Bell states under $SU(2)$. This is determined solely by the triplet component of the state. The only possibilities are for the state to have a residual symmetry of (a) a cylinder, (b) a double-sided rectangle, or (c) a sphere. This `shape' can be used to provide superselection rules when using the two-qubit state to simulate quantum operations. Details are given in section \ref{twoqubits}.}
\label{fig:friendlypicture}
\end{figure}

\section{Introduction}

Symmetry principles are ubiquitous in both the classical and quantum realms. They constrain dynamics, connect with conservation laws and simplify computations of physical properties. In crystallography the crystal structure is described by a symmetry group and largely determines the properties of the material \cite{tinkham}.

Symmetry considerations have been extremely successful in the field of quantum information theory \cite{definetti, hiddensubgroup, entanglement}. Moreover a substantial component of this work has addressed the regime in which symmetry principles disconnect from conservation laws \cite{extendingnoether}, and which has found application in the study of quantum features of thermodynamics \cite{timetranslation,gour,matteonature}.  Quantum states that break the symmetry can also be used to circumvent limitations on the precision of measurements imposed by conservation laws, as described by the WAY theorem \cite{wigner, arakiyanase, LoveridgeBusch, way, marvianway, popescuway}, while metrology can be viewed as using broken symmetries to distinguish different group transformations \cite{metrology}. Symmetry principles have also provided new insights into quantum speed limits and how quickly quantum operations can be performed \cite{speedlimits}. Finally, symmetry groups provide an invaluable toolkit in the context of quantum computing \cite{stabilizer, gottesmanthesis, benchmarking, nielsen}.

In this work we look at how general quantum states $\rho$ can be classified by groups that break a given symmetry constraint to an equal degree. A concrete example makes this clearer: both a chair and an arrow break rotational symmetry, however an arrow has a residual (axial) symmetry group to it, and so breaks rotations to a weaker degree than the chair. For a fixed quantum system of dimension $d$, and symmetry action $G$, we can ask what residual symmetries does the system admit, and how do these residual symmetries transform under symmetric dynamics. We also address the degree to which this aspect of a quantum system constitutes a property in the `resource-theoretic' sense.

The structure of the paper is as follows. In the next section we provide notation and background motivation. In section \ref{isotropy} we provide the coarse-graining scheme and some basic relations involved. This demonstrates the complexity of the topic, and highlights some subtleties if one wishes to associate the scheme to a quantum resource. In section \ref{twoqubits} we provide a complete classification for two-qubit systems (partly pictured in Fig. \ref{fig:friendlypicture} and more fully in Fig. \ref{fig:horodecki}), and which illustrates how quantum state spaces admit a coarse-grained partial order.

\section{Background and Motivation}

In this paper we shall use the following notation. For any quantum system $A$ we denote by $\H_A$ its associated Hilbert space,  and $\B(\H_A)$ the set of linear operators on $\H_A$. We assume that a symmetry action of $G$ is defined on $\H_A$ through the unitary representation $g \mapsto U(g) \in \B(\H_A)$. The group action on a quantum state $\rho \in \B(\H_A)$ is given by the adjoint action $\mathcal{U}_g (\rho) := U(g) \ \rho \ U^\dagger (g)$. 

For any quantum channel $\E: \B(\H_A) \rightarrow \B(\H_{A'})$, the action of the group on $\mathcal{E}$ is $\mathfrak{U}_g (\E) := \mathcal{U}_g \circ \mathcal{E} \circ \mathcal{U}_{g^{-1}}$ (where $\circ$ denotes the concatenation of channels). This can be interpreted as moving to a different frame, performing the channel and returning to the original frame. 

If $\mathfrak{U}_g (\E) = \mathcal{E}$ for all $g\in G$ then $\E$ is called a \emph{symmetric operation}. This includes certain preparation procedures, and the resulting states are known as \emph{symmetric states} because they are invariant under the group action: $\mathcal{U}_g(\rho) = \rho$ for all $g \in G$.

\subsection{General state transformations under a symmetry constraint}\label{interconversion}

A primary concern here is the study of general quantum states or quantum operations under a symmetry action. In particular we would like to study the degree to which any given quantum state breaks the symmetry (a `resource state' \cite{vaccaroasymmetry, frameness}). One way of comparing states in this manner is to say that $\rho \succ \sigma$ if and only if $\rho \rightarrow \sigma = \E(\rho)$ for some \emph{symmetric} quantum operation $\E$, which defines a partial order on quantum states \cite{resourceframes, asymmetry}. A \emph{measure} of how much a state breaks the symmetry is then a function $M$ that respects the partial order -- namely, if $\rho \succ \sigma$ then $M(\rho) \ge M(\sigma)$ \cite{vaccaroasymmetry, frameness}. A central question is whether a \emph{complete} set of measures exist that fully capture this ordering of states, and thus asks when two quantum states can or cannot be interconverted under the symmetry. Very recently \cite{gour} an answer was given for this general problem in terms of single-shot entropies of the form $H_{\mbox{\tiny min}} (R|A)$ that quantify the amount of information in a system $A$ about an external quantum reference frame $R$. Thus it is the case that $\rho \rightarrow \sigma$ under the symmetry constraint if and only if $H_{\rm min} (R|A)$ decreases with respect to all reference frames $R$. 

However this set of measures for the resource theory is over-complete and highly complex to use. Applying it to general systems requires sophisticated techniques in single-shot information theory. Here we wish to tackle a more modest goal, and instead of describing the complete symmetry properties of a quantum state, wish to coarse-grain states into groups that break the symmetry in the same way. Notably, this itself turns out to be a highly non-trivial problem and so sheds light on the structure of the more general set of measures given in \cite{gour}.

\subsection{Efficient simulation of quantum channels} \label{local}

One very grounded way of probing the degree to which a quantum state breaks a symmetry is to see how useful it is when we wish to perform tasks, such as induce a channel on another system. More generally, if we have some target quantum operation $\E$ on a system, we would like to know what kind of resource states and interactions are required to realise that operation. The degree to which a resource $\sigma$ can do this under symmetric dynamics therefore constitutes a measure of its properties.

In \cite{cristina}, the symmetry properties of quantum channels were analysed, with a focus on the orbit of a channel $\E$, defined as 
\begin{equation}
\mathcal{M}(\mathcal{E}) := \{ \mathfrak{U}_g(\E) : \ g \in G \}.
\end{equation}
A similar definition applies to the orbit of states under the group, so $\M(\rho) = \{ \U_g(\rho) : \ g \in G \}$.

Now the type of orbit one obtains provides a natural context in which to understand both how the channel breaks a symmetry constraint, and moreover how one might simulate such a channel using an auxiliary quantum system $B$ prepared in a non-symmetric state $\sigma_B$. 

It is clear that an auxiliary resource state may be used with varying efficiency in the simulation of a channel on a system $S$, however in \cite{cristina} the question of when a state $\sigma_B$ can be used to induce a channel $\E$ such that its ability to simulate the very same channel on subsequent systems is undiminished. This was motivated by the discovery of a ``catalytic coherence'' protocol \cite{aberg} in which a quantum state $\sigma_B$ is used as a phase reference. The state undergoes non-trivial changes $\sigma_B \rightarrow \sigma_B^{(1)} \rightarrow \sigma_b^{(2)} \rightarrow \cdots \sigma_B^{(n)}$ yet for all $n \ge 1$ its abilities as a phase reference never diminish. In \cite{cristina} the structure of such \emph{arbitrarily repeatable} protocols was spelled out using harmonic analysis. It was shown that if a channel $\E$ on $S$ is simulated using a quantum state $\sigma_B$ under an arbitrarily repeatable protocol then there exists a POVM measurement $\{M_k\}$ on $B$ such that
\begin{equation}
\E(\rho) = \sum_k \tr(M_k\sigma) \Phi_k (\rho)
\end{equation}
where $\{\Phi_k\}$ are completely positive maps on the primary system $S$. The POVM $\{M_k\}$ on the auxiliary system depends explicitly on the target channel $\E$. Crucially however, it can be understood as resolving the location $\x$ of $\E$ on its unitary orbit $\M(\E)$ to some scale that depends on the dimension $d$ of $S$. For example, in the case of catalytic coherence, the POVM measurement on $B$ corresponds simply to resolving a point $\x$ on a circle (the orbit of the channel $\E$ under phase rotations) down to an angular scale $\sim \frac{2\pi}{d}$. 

The degree to which $\x$ is resolved on $\M(\E)$ clearly depends on the resource state $\sigma_B$, and corresponds to the fact that $\sigma_B$ is a \emph{quantum} reference frame. This in turn can be viewed as inducing non-classical geometry on $\M(\E)$ -- for example, for $G=SU(2)$ a channel can have orbit $\M(\E)$ that is a 2-sphere. The use of a bounded reference system in this case means that only finite resolution of points on $\M(\E$) is possible, however since the 2-sphere carries a phase space structure the bounded reference case corresponds to the so-called ``fuzzy sphere'' model from non-commutative geometry \cite{fuzzysphere}. This differs from catalytic coherence in that we now have complementarity in resolving the coordinates on the sphere, and which ultimately arises because the two orbits have quite different structure.

Thus the general efficient use of a quantum state $\sigma_B$ to simulate a channel $\E$ on $S$ can be analysed in terms of two aspects:
\begin{enumerate}
\item (`Shape') The type of orbit $\M(\sigma)$ the state has under the group action compared to $\M(\E)$, independent of metrical aspects.
\item (`Geometry') The ability of $\sigma$ to encode classical coordinate data for $\x$ in $\M(\E)$.
\end{enumerate}

\noindent What is a necessary relation between $\M(\E)$ and $\M(\sigma)$? It is clear that we need $\sigma$ to break the symmetry to a `larger' degree than $\E$, however in this paper we would like to make this statement more precise. Given the complexity of the general problem, one motivation in this work is to start with the above separation into `shape' and `geometry', and develop an organisational setting in which we provide a coarse-graining over states of the same `shape', with the remaining task being to analyse the `geometry' aspect of the state.

\subsection{Related topics in quantum state tomography and quantum computation}

Beyond the abstract problem of simulating an arbitrary quantum channel using a resource state, there are specific contexts of importance where such a coarse-grained division naturally arises. In \cite{tomography}, quantum state tomography was studied in the context of prior information that restricts the state to a lower-dimensional submanifold of the state space. The analysis shows that the topological genus of the manifold can be used to bound the number of measurements needed to discriminate states on the submanifold. The orbit of a state (or channel) is one such submanifold, and so the study of what kinds of orbits can exist and how they transform among themselves is of potential relevance to quantum tomography under symmetry constraints.

A wholly practical direction where such analysis might be of relevance is in the computational power of gate-sets in quantum computation, and how such gate-sets interact with states that increase the computational power of the gate-set (such as magic states for the stabilizers \cite{magic}). For example, a recent work provides a classification of all Clifford gate-sets \cite{stabilizer} in a lattice hierarchy, and so aspects of the present work might be applicable in studying how noisy quantum states increase the computational power of easily realisable gate-sets.

\subsection{Core questions of the present work}

In subsection \ref{interconversion} we highlighted that state interconversion $\rho \rightarrow \sigma$ under a symmetry constraint is highly non-trivial, and while progress has been made on this fundamental problem, it is not clear how complex it is in general. Also, as discussed in subsection \ref{local} recent results on optimal protocols show that the structure of this theory naturally decomposes into a geometric component and the particular kind of group orbit the quantum channel has.
 In light of these considerations one expects that the group orbit level of description might provide a form of coarse-grained description that sheds light on the more complex parent resource theory. By studying this coarse-grained theory we can shed light on the parent theory, and moreover extend the notion of a resource theory from being about pre-orders on quantum states to considering pre-orders on \emph{sub-sets} of states. However, in order for this to work and be meaningful, it is crucial that the coarse-grained theory be (a) physically sensible and (b) relate naturally to the parent theory. Therefore the core questions we address in this work are the following:
\begin{enumerate}
\item Does a coarse-grained resource-theory framework exist? 
\item Is this theory consistent with its parent resource theory? Is it physically sensible?
\item How complex is the general structure of symmetry-based resource theories? 
\end{enumerate}
In the next section we tackle the first two of these questions, by showing how a natural partition of the state space arises that is consistent with the finer-grained resource theory, and describe its limitations and physical robustness. In the process, and through exhaustive analysis of the $G=SU(2)$ case for two qubits, we shed light on the crucial third question posed.

\section{Coarse-graining states and channels under a Symmetry Action} \label{isotropy}

\subsection{Simulating Operations Under Symmetry Constraints}

If one wishes to realise a quantum channel $\E$ via a symmetric interaction with a state $\sigma$, the most natural way to model this is as
\begin{equation}
\E( \rho ) = \tr_B V( \rho_A \otimes \sigma_B)V^\dagger
\end{equation}
where $V$ is a covariant unitary, namely $[V, \ U_A(g) \otimes U_B(g)] = 0$ for all $g \in G$. If such a relation holds, we say that $\E$ can be \emph{simulated} using the state $\sigma_B$.

A basic classification of states and channels under a symmetry can be done by studying the subgroup of residual symmetries for the state or channel, called the isotropy subgroup or stabilizer of $\E$ and defined as
\begin{equation}
\rm{Iso}(\E) := \{ g \in G \ : \ \mathfrak{U}_g(\E) = \E \}.
\end{equation}
It specifies the residual symmetry that the channel has under the group action. Symmetric channels are invariant under all group transformations, and simply have $\rm{Iso}(\E)=G$. This notion also applies to states, where $\rm{Iso}(\rho) := \{ g \in G \ : \ \mathcal{U}_g(\rho) = \rho \}$.

In many cases the quantum operation $\E$ has some residual symmetry, which is characterised an isotropy subgroup $H$ of $G$. The possible isotropy subgroups form an abstract structure called a \emph{lattice} \cite{subgrouplattice}. The partial ordering of the subgroups is defined by subset inclusion, namely $H_1 \prec H_2$ if $H_1 \subseteq H_2$ for $H_1, H_2 \subseteq G$. In addition, every pair of elements have a unique supremum and unique infimum, which define binary operations called the meet and join \cite{lattice}. For two subgroups $H_1$ and $H_2$, their meet is denoted $H_1 \wedge H_2$ and defined as $H_1 \cap H_2$, while their join is denoted $H_1 \vee H_2$, defined as the subgroup \emph{generated} by $H_1 \cup H_2$ \cite{subgrouplattice}. This structure is illustrated by a Hasse diagram of the subgroups, as shown in Fig. \ref{fig:csubh}.

\begin{figure*}[t]
\includegraphics[width=16cm]{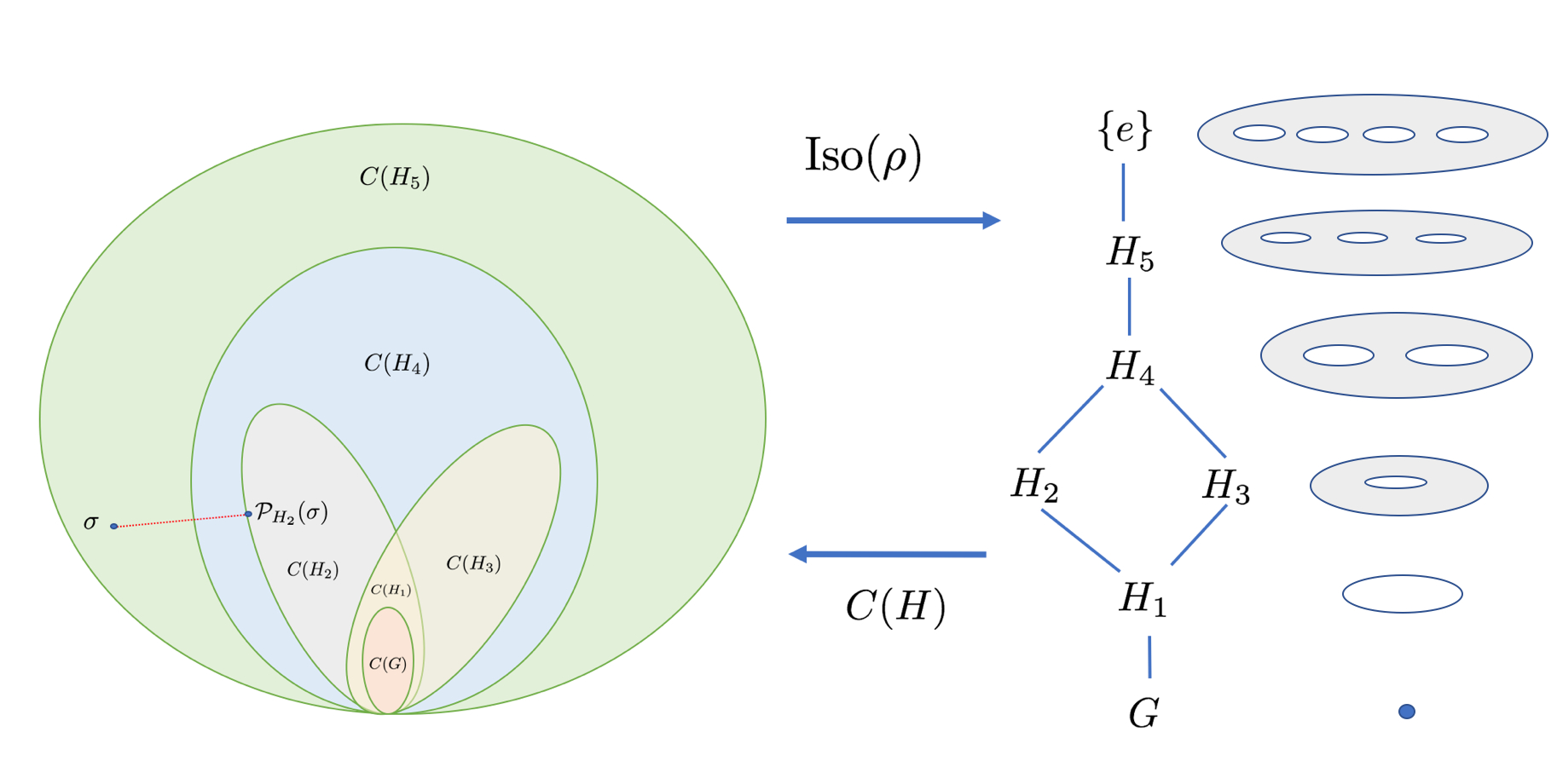}
\caption{\textbf{Coarse-grained classification of quantum states under a symmetry.} The set of all density operators is partitioned into sets $C(H)$ with definite symmetries (seen on the left). The $C(H)$ correspond to the subgroups shown in the Hasse diagram in the middle. The Hasse diagram represents the subgroup lattice, with lines indicating subgroup inclusion, i.e. $H_5 \prec H_4 \prec H_3$ etc. Note that for finite-dimensional systems many of the $C(H)$ will be empty. The subsets $C(H)$ can be associated, up to diffeomorphisms, with group orbits (indicated on the right). The left figure shows the action of group-averaging (over $H_2$), going from $\sigma$ to $\mathcal{P}_{H_2}(\sigma)$ as described in section \ref{projecting}. } \label{fig:csubh}
\end{figure*}

Since quantum operations can be combined in a number of ways, we present some basic statements that constrain the residual symmetry of the resultant quantum operation. The proofs for this section are provided in Appendix \ref{proofs}.

\begin{lemma} \label{concatenate}
Given any two quantum channels $\E:\B(\H_A) \rightarrow \B(\H_{A'})$ and $\F : \B(\H_B) \rightarrow \B(\H_{B'})$, with unitary representations of a group $G$ defined on all input and output spaces. Then we have the following:
\begin{enumerate}
\item $\rm{Iso}( \E\otimes \F) \succ \rm{Iso}(\E) \wedge \rm{Iso}(\F)$ under the tensor product group action $\mathfrak{U}_g(\E \otimes \F) = \mathfrak{U}_g(\E) \otimes \mathfrak{U}_g(\F)$.
\item If $B' = A$ then $\rm{Iso}(\E \circ \F)  \succ \rm{Iso}(\E) \wedge \rm{Iso}(\F)$.
\item If $A=B$ and $A'=B'$, and $p$ some probability $0\le p \le 1$, then $\mathrm{Iso}( p \E + (1-p) \F) \succ \mathrm{Iso}(\E) \wedge \mathrm{Iso}(\F)$.
\item $\mathrm{Iso}(\U_g \circ \E \circ \U_g^\dagger) = g[ \mathrm{Iso}(\E)]g^{-1}$. 
\end{enumerate}
\end{lemma}

\noindent These results also apply for quantum states, once we view the state as the state preparation map $1 \rightarrow \rho$. We can now make precise a coarse-grained (necessary, but far from sufficient) requirement that a state $\sigma$ allow the simulation of a channel $\E$ under symmetric dynamics and show that this bound is tight in general. The proof is provided in Appendix \ref{proofs}.

\begin{theorem} \label{simulate}
If a system $B$ in a state $\sigma_B$ can be used to simulate a CPTP map $\E$ under symmetric dynamics, then $\text{Iso}(\sigma_B) \prec \text{Iso}(\E)$. Moreover, if $\E$ has isotropy group $Iso(\E)$ then there exists a quantum system $B$ and quantum state $\sigma_B$ with $\text{Iso}(\sigma_B) = \text{Iso}(\E)$ that allows the simulation of $\E$.
\end{theorem}

\noindent This result gives the basic relation, in the coarse-grained picture, for when a quantum state $\sigma_B$ can be used to simulate a quantum channel and also shows that this relation is a tight one. This is necessary in order for the isotropy subgroup approach to be consistent with the parent resource theory, where one is concerned with the minimal resources needed to realise a particular quantum channel.

A similar result can be given in the case of approximate simulation.  For any target operation $\E$ we can define the noisy version given by
\begin{equation}
\E_\epsilon = (1-\epsilon) \ \E + \epsilon \ \D
\end{equation}
where $\D$ is the complete depolarization map $\rho \mapsto \frac{1}{d}\I$ for all $\rho$. From Lemma \ref{concatenate} we see that $\mbox{Iso}(\E_\epsilon) = \mbox{Iso}(\E) \wedge G = \mbox{Iso}(\E)$. If $\sigma$ allows the simulation of $\E_\epsilon$ for some $0< \epsilon <1$ then this corresponds to an approximate, `isotropic' simulation of the original map with noise parameter $\epsilon$. Theorem \ref{simulate} extends to this case in an obvious way.

\subsection{Is `shape' a resource-theoretic property?}

While we normally associate measurable properties with either projective measurements or more generally POVMs, there is an alternative way that is more general again. Specifically, a property is associated with a \emph{pre-order} defined on the set of all quantum states. This recent approach is called the resource-theory method, and has found success in areas such as entanglement, coherence, thermodynamics, and many other scenarios \cite{resourcetheory,entanglement, coherencereview, thermoreview}.  Quantum maps that respect the pre-order are called `free operations' and any real-valued function on quantum states that respects the pre-order is a \emph{measure} of the property.

Given an ability to order the isotropy subgroups of quantum states and channels under the group action, we can ask if a meaningful notion of `shape' can be defined for quantum systems, along such a resource-theoretic line. This would essentially characterise the asymmetry of states and channels without reference to a measure.

For continuous Lie groups, we find that the orbit of a state or channel is always a homogeneous space \cite{karol,homogeneous}, specified by both the group $G$ and the particular isotropy subgroup $H$ of the channel. More precisely the orbit of the channel is $\M(\E) $ and coincides with the quotient $G/H$ up to diffeomorphisms, which we write $\M(\E) \cong G/H$ (and likewise for states).\footnote{Technically we define the `shape' of the group orbit with respect to the conjugacy class of the isotropy subgroup $H$. This `shape' is also called the Orbit-Type in the literature. See \cite{bredon} for further details.} This is what we call the `shape' of a state or channel.

A simple example can be given for a single qubit state under an $SU(2)$ symmetry, where the possible types of group orbit are
\begin{equation}
\mathcal{M} \left(\frac{1}{2}[ \mathds{1} + \mathbf{r} \cdot \sigma ]\right) =
\begin{cases}
SU(2)/U(1) \cong S^2 & \mathbf{r} \neq 0 \\
SU(2)/SU(2) \cong \{e\} & \mathbf{r} = 0,
\end{cases}
\end{equation}
and these are illustrated in Fig. \ref{fig:singlequbit}.

\begin{figure}[h]
\includegraphics[width=8cm]{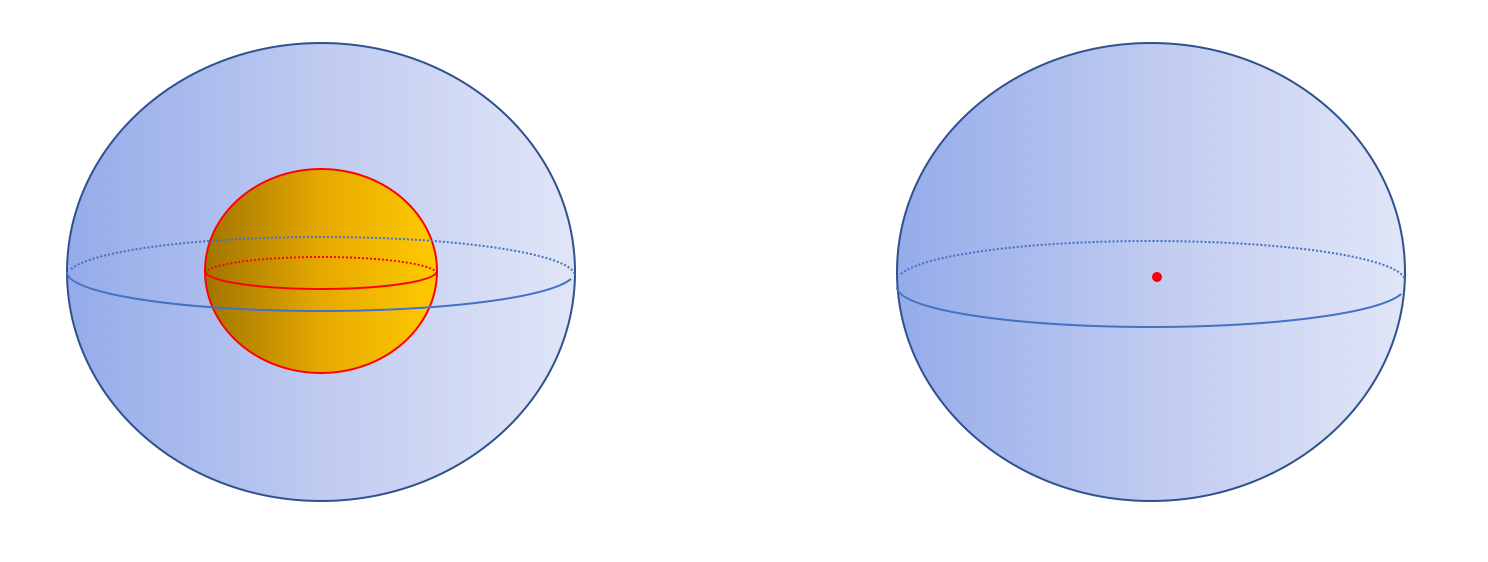}
\caption{\textbf{Basic orbits of states.} Group orbits for a single qubit under the action of $SU(2)$ shown in the Bloch sphere. There are only two possibilities in this case: a sphere for non-zero Bloch vectors (left), or a point for the symmetric maximally mixed state (right).}
\label{fig:singlequbit}
\end{figure}

The lattice of isotropy subgroups gives us a way of comparing the `shapes' of group orbits under a group action $G$, with $\M(\E) \prec \M(\F)$ iff $\rm{Iso}(\E) \succ \rm{Iso}(\F)$ (up to isomorphism of the isotropy subgroups). Likewise for states, $\M(\rho) \prec \M(\sigma)$ iff $\rm{Iso}(\rho) \succ \rm{Iso}(\sigma)$ (up to isomorphism). The convention here is to match with resource-theoretic measures of asymmetry \cite{resourceframes, vaccaroasymmetry}, so that one `shape' is `bigger' than another if it has a smaller isotropy group. For example, a single-qubit state with non-zero Bloch vector ($\M(\rho) \cong S^2$) has a `bigger' group orbit than the maximally mixed state ($\M(\mathds{1}/2) \cong \{ e \}$). This gives valuable information regarding the resources required to perform a quantum channel in the presence of symmetry constraints.

The isotropy subgroup gives us a natural equivalence relation $\sim$ between states that behave in the same way under the group action. We can then say that $\rho_1$ and $\rho_2$ are related $\rho_1 \sim \rho_2$ if we have that $\rm{Iso}(\rho_1) = \rm{Iso}(\rho_2)$. This equivalence relation partitions the state space into sets $C(H) := \{ \rho  : \ \rm{Iso}(\rho) = H \}$, and collects together all states that break the symmetry in the same manner. Taking the union of $C(H)$ corresponding to isomorphic $H$ partitions the states and channels according to their `shape'.

Viewed as a mapping, $C$ maps from the subgroup lattice onto a partition of the state space, and so we can simply allow the partition to  inherit the lattice structure, and order subsets of states as $C(H_1) \succ C(H_2)$ whenever $H_1 \prec H_2$.

Note however that the sets $\{C(H)\}$ are not convex in general. For example a qubit system under $G=SU(2)$, both $\ket{0}\bra{0}$ and $\ket{1}\bra{1}$ have the same $U(1)$ isotropy subgroup, but $\frac{1}{2}(\ket{0}\bra{0} + \ket{1}\bra{1}) = \frac{1}{2} \mathds{1}$, which is symmetric. However, $C(G)$ is always a convex set. This follows because if $\rho_1, \rho_2 \in C(G)$ then Lemma \ref{concatenate} implies that $G \succ \rm{Iso}(p_1\rho_1 + p_2 \rho_2) \succ \rm{Iso}(\rho_1)\wedge \rm{Iso}(\rho_2) = G$, and so any mixture has the same isotropy subgroup.

The $SU(2)$ qubit example also highlights that if $C(H) \neq \emptyset$ it does not imply that $C(H') \neq \emptyset$ for all $H' \succ H$, since there are subgroups of $SU(2)$ containing $U(1)$ which do not appear as the isotropy subgroup of a single qubit state. Not all possible isotropy subgroups will be seen in a given system, and for a finite dimensional system many of the sets $C(H)$ will be empty.

However in order to have a meaningful resource-theory interpretation, the coarse-grained ordering must be compatible with the set of free operations -- namely the symmetric operations. It is easy to show that this is in fact the case at the level individual quantum states.
\begin{lemma} \label{isotropy nondecreasing}
Under a symmetric operation $\E$, $\mathrm{Iso}( \E (\rho)) \succ \rm{Iso}(\rho)$.
\end{lemma}

\noindent This shows that the coarse-grained ordering of states is consistent with the set of free (symmetric) operations, as expected. However it does not imply that a functional mapping is defined on the sets $C(H)$; it is possible to have $\rho_1$ and $\rho_2$ in the same set $C(H)$, but each get sent to different sets $C(H'_1)$ and $C(H'_2)$. This aspect means that interpreting the coarse-graining as defining a quantum resource needs care.

\subsection{The speck of dust argument -- tiny perturbations to symmetric states \& channels}
	
Having outlined how the set of quantum states or the set of quantum channels of a quantum system is partitioned according to the symmetry action one might raise the following concern: what happens if one perturbs a channel via some small perturbation? How does this affect its isotropy subgroup?

It is readily seen that one can make an arbitrarily small perturbation to any quantum channel $\E$ so as to break the symmetry in a ``maximal'' way. Informally this amounts to the ``speck of dust'' argument where, for example, a perfectly rotationally symmetric ball can be made completely asymmetric by sticking arbitrarily small pieces of dust onto it (Fig. \ref{fig:speck}).

\begin{figure}[h]
\includegraphics[width=8.5cm]{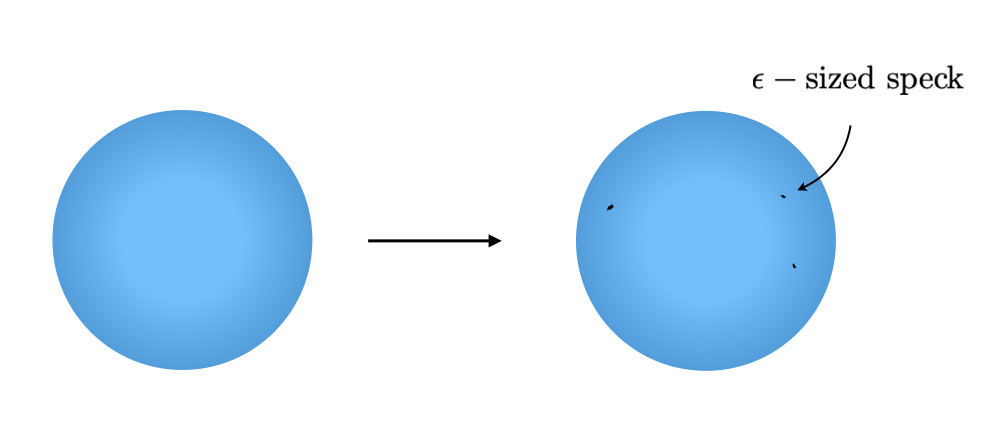}
\caption{\textbf{The speck of dust argument.} If we start with a perfectly rotationally symmetric object (left image) we can always add arbitrarily small perturbations that completely break the rotational symmetry (right image). However physically the fact that we have finite resolution means that quantum states and channels should always be understood as being defined up to some $\epsilon$--smoothing scale. This does not affect the resource-theoretic account and therefore the `shape' is a robust feature in the presence of such tiny symmetry-breaking perturbations. }
\label{fig:speck}
\end{figure}

In more abstract terms, one can use the principal orbit-type theorem \cite{bredon} for the space of quantum channels of a system to see that there is always a \emph{principal isotropy} subset that forms a dense subset in the set of all quantum channels. Thus ``most'' quantum channels on a system will break a symmetry to a maximal degree (the principal isotropy). 

This argument of course applies equally well to the partitioning of the quantum state space into different subsets with fixed isotropies, however in both cases this does not pose a problem for our formalism for the following reasons. Firstly, in resource theories we are fundamentally interested in the most efficient use of resources and therefore our task is to look for those quantum resource states that break the symmetry not in a maximal way, but in a \emph{minimal} way. 

Secondly, at the level of channels, the most interesting and commonly addressed channels in quantum information science are \emph{not} of principal isotropy type — namely they have non-trivial isotropies. For example, consider quantum channels on a single qubit, and the irreducible group action of $G=SU(2)$. A typical quantum channel on the qubit will only have a small residual isotropy group $\mathbb{Z}_2 = \{ \mathds{1}, - \mathds{1} \}$ (the principal isotropy subgroup of the set of all qubit channels). However consider the following standard single qubit channels:
\begin{enumerate}
\item The depolarizing channel: $\E(\rho) = p \rho + (1-p) (\frac{1}{2} \I)$, with $0\le p \le 1$.
\item A rotation about the $\hat{n}$ axis: $U_{\hat{n}} (\theta) = \exp [-i\hat{n}\cdot \boldsymbol{\sigma}]$.
\item A projective measurement along the $\hat{n}$ axis: $\E(\rho) = \sum_k \tr(\Pi_k \rho) \Pi_k$, where $\Pi_0 = \ket{\hat{n}}\bra{\hat{n}}$ and $\Pi_1 = \mathds{1} - \Pi_0 =  \ket{-\hat{n}}\bra{-\hat{n}}$.
\item Partial dephasing along the axis $\hat{n}$: $\E(\rho) = p \rho + (1-p)  \sum_k \tr(\Pi_k \rho) \Pi_k$, where $\Pi_0 = \ket{\hat{n}}\bra{\hat{n}}$ and $\Pi_1 = \mathds{1} - \Pi_0 =  \ket{-\hat{n}}\bra{-\hat{n}}$.
\item A qubit state preparation channel: $\E(\rho) = \frac{1}{2} ( \mathds{1} + p \ \hat{n} \cdot \sigma )$.
\end{enumerate}

\noindent It is readily seen that the isotropy subgroup for channel (1) is $SU(2)$, while for channels (2--5) it is a $U(1)$ subgroup of $SU(2)$. The orbits will have `shapes' $SU(2) /SU(2) \cong \{ e \}$ and $SU(2) /U(1) \cong S^2$ respectively. None of these channels have principal isotropy type $\{\I , -\I\}$, however they are clearly key quantum operations in quantum information science for which one might wish to determine the minimal quantum resources necessary to realise them. 

Thirdly, the fact that something is “measure zero” in some space does not imply that it is physically irrelevant. For example, a 2-d plane is a measure zero set in three spatial dimensions, however this does not mean that anyonic physics or the quantum hall effect is irrelevant. Or of closer bearing on our present work, we can consider the case of Noether’s theorem: clearly the set of dynamics with rotational symmetry is again “measure zero”, however this does not mean that Noether’s theorem is irrelevant for physics. In both cases while the restrictions are strictly unstable under small perturbations, the important thing is that the features of interest are operationally robust under small perturbations. 

The speck of dust subtlety makes itself apparent in our setting if one tries to naively exploit some metric on quantum states to quantify how far apart the sets $\{C(H)\}$ are from each other.  It is readily seen that the distance been any two such sets is in fact zero.
\begin{lemma}
Let $d(\cdot, \cdot)$ be any metric on the space of quantum states. In terms of this metric we define
\begin{equation}
\mathbf{d}(C(H_1) , C(H_2)) := \inf_{\substack{\sigma_1 \in C(H_1) \\ \sigma_2 \in C(H_2)}} d(\sigma_1 , \sigma_2).
\end{equation}
Then $\mathbf{d}(C(H_1),C(H_2)) = 0$ for all $H_1, H_2 \prec G$. 
\end{lemma}

\noindent This also shows that any set $C(H)$ is arbitrarily close to the set $C(G)$. However, this simply means that any symmetry properties must always be considered up to some \emph{finite resolution scale} -- arbitrarily small perturbations can always eliminate residual symmetries, even though this state is practically indistinguishable from the unperturbed one. 

Therefore any membership of a quantum state $\rho$ to a subset should only be considered up to some smoothing scale $\epsilon$ based on a distance measure $d$ (such as from the $L_1$ norm). For each state there is an $\epsilon$-ball of nearby states, $\mathcal{B}_\epsilon (\rho) = \{ \sigma : d(\rho,\sigma) \leq \epsilon \}$. Rather than $\rm{Iso}(\rho)$, we should consider $\rm{Iso}_\epsilon (\rho) := \max_{\sigma \in \mathcal{B}_\epsilon (\rho)} \rm{Iso}(\sigma)$. Intuitively, if a state $\rho$ is within a distance $\epsilon$ of another state with more residual symmetries, we associate those additional residual symmetries with $\rho$. The smoothing scale selects the size of symmetry-breaking perturbations that we wish to consider, and will depend on the particular physical context.

This reasoning applies equally to quantum channels. For example, the depolarisation channel on a qubit is fully invariant under the $G=SU(2)$ and so no non-trivial resource state is needed to simulate it. If one perturbs this channel to $\E \rightarrow \E_\epsilon$ by some $\epsilon >0$ perturbation in the diamond norm one can clearly break this isotropy to the principal isotropy $\{\I,-\I\}$, however it is also clear (e.g. from the continuity of the Stinespring dilation \cite{stinespringcontinuity}) that $\E_\epsilon$ can either be simulated exactly with a resource state on $B$ that $\epsilon$-close to be being symmetric, or simulated approximately to the same threshold with a perfectly symmetric state on $B$. Thus the use of isotropy groups within the resource theory context is robust under such small perturbations and the speck of dust argument is moot. We illustrate this point explicitly in Section \ref{twoqubits}, where we demonstrate that for the case of two qubits under $SU(2)$ the notion of different `shapes' is perfectly robust under the particular case of perturbations that allow a $4\%$ error rate.

\subsection{Projecting out residual symmetries via quantum operations}\label{projecting}

Given the structure of states under the above partition, we can ask how easy it is to move from one set $C(H)$ to another. As discussed one does not in general have quantum operations that map any $C(H)$ neatly into some other subset. Instead it makes sense to consider the sets $\hat{C}(H) := \bigcup_{W \succ H} C(W)$. These sets are quite natural to consider because Lemma \ref{isotropy nondecreasing} implies that each $\hat{C}(H)$ is closed under symmetric operations, and any symmetric operation provides a well-defined mapping of $\hat{C}(H) \rightarrow \hat{C}(f(H))$.

The quantum operation
\begin{equation}
\mathcal{P}_H(\rho) := \int_H dh \ \mathcal{U}_h (\rho)
\end{equation}
is the average of the state $\rho$ over a fixed subgroup $H$ weighted by the invariant Haar measure $dh$. In the case of a finite subgroup of size $|H|$, the integral $\int_H dh$ is replaced by the sum $\frac{1}{|H|} \sum_{h \in H}$.
\begin{lemma} \label{H-Twirling}
The map $\mathcal{P}_H$ has the following properties:
\begin{enumerate}
\item $\mathcal{P}_H$ is the (orthogonal) projector onto $\hat{C}(H)$.
\item $\mathcal{P}_H(\rho) = \arg \min_{\sigma \in \hat{C}(H)} S( \rho || \sigma )$, where $S( \rho || \sigma )$ is the relative entropy.
\end{enumerate}
\end{lemma}
\noindent Therefore $\mathcal{P}_H$ projects onto $\hat{C}(H)$, as illustrated in Fig \ref{fig:csubh}. Dephasing \cite{nielsen}, $\mathcal{E}(\rho) = \frac{1}{2\pi} \int_0^{2\pi} d\theta \ e^{i \theta Z} \rho e^{-i\theta Z}$, can be viewed as $\P_{U(1)}$, sending $\rho$ to the nearest incoherent state while preserving the diagonal terms of a density matrix.

The mapping $\mathcal{P}_H$ moves down chains of the subgroup Hasse diagram, however this is not always a symmetric operation. The following tells us when such a transformation can be performed freely in the resource theory.
\begin{lemma} \label{twirling resources}
Given a group action for $G$,
$\rm{Iso}(\mathcal{P}_H) = N_G(H)$, where $N_G(H) = \{g \in G \ : \ gHg^{-1} = H \}$ is the normalizer \cite{normalizer} of $H$ in $G$, and therefore if $H$ is a normal subgroup of $G$ ($H \triangleleft G$) then $\P_H$ is a symmetric operation.
\end{lemma}
\noindent This constrains the kind of resources needed to move from one $C(H)$ to another using $\mathcal{P}_H$, since this projects onto a given $\hat{C}(H)$, although less resource-hungry ways may exist to perform the same transformation.

\subsection{Discussion of the section results}

In this section we have given an analysis of the what happens when one classifies quantum states or channels in terms of their residual symmetries. We discussed how composition and mixing affect the ordering, and how one can relate these features to the issue of simulation (exact or approximate) of a target quantum operation with some resource state. We also saw that the ordering has subtleties if one wishes to interpret it in a resource theoretic sense. The statements one makes are also quite blunt, and a good example of this is the fact that the sets $\{C(H)\}$ are all arbitrarily close to one another. However the relations still carry non-trivial content, and for example Theorem \ref{simulate} can be viewed as a form of superselection rule on quantum operations that tells us which states are ruled out and which are not.

While this is conceptually neat and provides high-level insight into the complexities of the full classification of states (as described in \cite{gour}), it is less clear how computationally useful or simple these structures are in practice. To address this point, in the next section we consider the important case of a two-qubit system under an $SU(2)$ action. We find that already in just this simple scenario the hierarchy of states is quite complex, however the example does provide insight into what to expect in the more general case.

\section{Illustrative Example: Two Qubits Under $SU(2)$ Symmetry Constraints}\label{twoqubits}

The goal of this section is to illustrate the classification of quantum states in a simple quantum system that has sufficient structure, yet is tractable to the point of being fully solvable. The state space of a two-qubit system is 15 dimensional and is sufficiently non-trivial, moreover there is a very natural group action to consider, namely the tensor product representation of $SU(2)$. The orbits of 2-qubit states have been studied in relation to thermodynamics and correlations within these states \cite{jevtic2012quantum, jevtic2012maximally}.

Any two qubit state $\rho_{AB}$ can be written
\begin{align}
\rho_{AB}&=\frac{1}{4} \bigg( \mathds{1} \otimes \mathds{1} + \mathbf{a} \cdot \sigma \otimes \mathds{1} + \mathds{1} \otimes \mathbf{b} \cdot \sigma  \nonumber \\
& \qquad \qquad \left. + \sum_{i,j = 1}^3 T_{ij} \ \sigma_i \otimes \sigma_j \right),
\end{align}
with local Bloch vectors $\mathbf{a}$ and $\mathbf{b}$ and correlation terms determined by the correlation matrix $T_{ij}$ \cite{tstates}. In general $|\mathbf{a}| \leq 1$ and $|\mathbf{b}| \leq 1$. Moreover, we can put all 2-qubit states into diagonal form,
\begin{align}
\rho_{AB} &= \frac{1}{4} \bigg( \mathds{1} \otimes \mathds{1} + \mathbf{a} \cdot \sigma \otimes \mathds{1} + \mathds{1} \otimes \mathbf{b} \cdot \sigma   \nonumber \\
&\qquad \qquad \left. + \sum_{i=1}^3 \tau_{i} \ \mathbf{c}_i \cdot \sigma \otimes  \mathbf{d}_i \cdot \sigma \right)
\end{align}
where $\{ \mathbf{c}_i \}$ and $\{ \mathbf{d}_i \}$ are orthonormal bases of $\mathbb{R}^3$.  The coefficients $\tau_i$ specify a position $(\tau_1, \tau_2, \tau_3)$ in a tetrahedron with vertices $(-1,-1,-1), (1,1,-1), (1,-1,1)$ and $(-1,1,1)$.  When $\mathbf{a} = \mathbf{b} = 0$, then $\rho_{AB}$ is a quantum state with maximally mixed marginals represented by a point in the tetrahedron. In the case $\mathbf{a}$ or $\mathbf{b}$ being non-zero, every quantum state $\rho_{AB}$ corresponds to a point in the tetrahedron, however the converse is not true: not all triples $(\mathbf{a},\mathbf{b}, \mathbf{t})$ correspond to a valid quantum states.

Via local unitaries $U_1 \otimes U_2$, these can be transformed into the canonical form
\begin{align}
\tilde{\rho}_{AB} &= \frac{1}{4} \bigg( \mathds{1} \otimes \mathds{1} + \mathbf{a}' \cdot \sigma \otimes \mathds{1} + \mathds{1} \otimes \mathbf{b}' \cdot \sigma  \nonumber \\
&\qquad \qquad \left. + \sum_{i=1}^3 \tau_{i} \ \sigma_i \otimes \sigma_i \right)
\end{align}
where $\mathbf{a}' \cdot \sigma = U_1 ( \mathbf{a} \cdot \sigma ) U_1^\dagger$ and $\mathbf{b}' \cdot \sigma = U_2 ( \mathbf{b} \cdot \sigma ) U_2^\dagger$. The states in this canonical form with maximally mixed marginals are called T-states \cite{tstates},
\begin{equation}
\rho_T = \frac{1}{4}\left(\mathds{1} \otimes \mathds{1} + \sum_j \tau_j \ \sigma_j \otimes \sigma_j \right).
\end{equation}
The set of T-states is the convex hull of the four Bell states. The corners of the tetrahedron defined by the coefficients $\{\tau_i\}$ correspond to the Bell states.

We now consider the symmetry properties of 2-qubit states under $SU(2)$ symmetry constraints, assuming an $SU(2)$ tensor product representation for the group action, i.e. $\mathcal{U}_g (\rho)= U(g) \otimes U(g) \ \rho \ U^\dagger (g) \otimes U^\dagger (g)$, where $U(g)$ is a two dimensional irrep of $SU(2)$.

Our description of the group orbit as the quotient space $G / \rm{Iso}(\rho)$ starts with the group manifold of $G$ and eliminates the redundancy caused by the symmetries of $\rho$. The group transformations in $\rm{Iso}(\rho)$ have trivial group action on $\rho$, so we quotient out the equivalence classes $g \ \text{Iso}(\rho) := \{ gh_i : \ h_i \in \text{Iso}(\rho) \}$.

Certain groups have structure which permits us to further simplify our description of the group orbit manifold. Here we consider $SU(2)$ symmetry constraints, however this $SU(2)$ action has a subgroup $\mathbb{Z}_2 = \{ \pm \mathds{1} \}$ that is always a symmetry of a quantum state, and which is associated with $SU(2)$ being the double cover of $SO(3)$. Moreover this sub-group has a key property which allows us to simplify our description of the $SU(2)$ isotropy groups to those of $SO(3)$. We can therefore think in terms of $SO(3)$ subgroups, which are the familiar chiral \emph{point groups} one encounters in for example crystallography \cite{tinkham}.

A subgroup $N$ of $G$ is a normal subgroup, denoted $N \triangleleft G$, if it satisfies $gNg^{-1} = N$ for all $g \in G$ \cite{tinkham}. These partition $G$ into equivalence classes $gN := \{ gn_i : \ n_i \in N \}$, which themselves constitute elements of the quotient group $G/N$. Every element of $G$ is specified by the pairing of the equivalence class $gN$ and an element of the normal subgroup $N$. Likewise, elements of $H$ are specified by the pairing of an element of $H/N$ and an element of $N$ because $N \triangleleft H$. This means that the group orbit manifold $G/H$ can be described by the quotient groups $G/N$ and $H/N$, and we can write $\M( \rho) \cong G/H \cong (G/N) / (H/N)$. In our present analysis $\mathbb{Z}_2$ is a normal sub-group of $SU(2)$ common to all quantum states and so $ \M(\rho) \cong SO(3) / (H/\mathbb{Z}_2)$ where $H$ is the isotropy group of $\rho$ in $SU(2)$. We can therefore use the more intuitive $SO(3)$ subgroups, listed in Table \ref{table:so3subgroups}.

\begin{table}[h] 
\begin{tabular}{ | c | c |} 
    \hline
    Group & Description  \\ \hline
    $C_n$ & Cyclic group of order $n$ \\
    $C_\infty$ & Symmetries of Cone \\
    $D_n$ & Symmetries of a regular $n$-sided Polygon  \\
    $D_\infty$ & Symmetries of Cylinder \\
    $T$ & Symmetries of Tetrahedron \\
    $O$ & Symmetries of Cube/Octahedron \\
    $I$ & Symmetries of Dodecahedron/Icosahedron \\
    \hline
\end{tabular}
\caption{\textbf{Chiral point groups: subgroups of $SO(3)$,} as detailed in \cite{tinkham}. All quantum states have isotropy subgroups under $SU(2)$ that can be related to particular chiral point group.} \label{table:so3subgroups}
\end{table}

\noindent In the next section we look at the isotropy subgroups of the Bell states, and then go on to classify the symmetry properties of the T-states. Then we continue onto the wider class of 2-qubit states with maximally mixed marginals and finally we complete the classification for general 2-qubit states by considering non-zero local Bloch vectors.

\subsection{Bell States}

The Bell states correspond to the extremal points of the tetrahedron of T-states. Their isotropy subgroups under the $SU(2)$ group action can be calculated directly, as described in Appendix \ref{Bell}.

The singlet Bell state $\psi^-$ is symmetric under $SU(2)$, therefore $\rm{Iso}(\psi^-) = SU(2)$ and so $\M(\psi^-) = \{e\}$ as expected. The triplet Bell states do break the $SU(2)$ symmetry, with
\begin{align}
\rm{Iso}(\phi^+) &= \{ (iZ)^\alpha e^{i \theta Y} : \ 0 \leq \theta < 2\pi, \ \alpha = 0,1,2,3 \} .
\end{align}
We will call this subgroup $K_\infty \cong U(1) \rtimes \mathbb{Z}_4$, and it provides an example of a semi-direct product group\footnote{ The semi-direct product \cite{tinkham} is constructed from two groups, in this case $\mathbb{Z}_4$ and $U(1)$. Each element of the semidirect product group can be thought of as a pairing $(h,n)$ where $h \in H = \mathbb{Z}_4$ and $n \in N = U(1)$. However the multiplication law of the semidirect product group does not treat the elements of this pairing independently, with $(h_1, n_1) (h_2,n_2) = (h_1h_2,f_{h_2}(n_1)n_2)$ where $f_{h} : N \rightarrow N$ is a automorphism on the second group specified by an element of the first group.  The semi-direct product $\rtimes$ is defined by the choice of automorphism $f_h$. Note that $N \triangleleft (N \rtimes H)$. If a trivial (identity) automorphism is chosen, where $f_h(n) = n$ for any $h$ and $n$, we have a direct product group \cite{tinkham}. In the $U(1) \rtimes \mathbb{Z}_4$ example, the group multiplication law is
\begin{equation*}
\left[(iZ)^{\alpha_1} e^{i \theta_1 Y} \right] \ \left[ (iZ)^{\alpha_2} e^{i \theta_2 Y} \right] = (iZ)^{\alpha_1 + \alpha_2} \ e^{i [(-1)^{\alpha_2} \theta_1 + \theta_2] Y}.
\end{equation*}
}.
Similarly,
\begin{align}
\rm{Iso}( \phi^-) = \{ (iY)^\alpha e^{i \theta X} : \ 0 \leq \theta < 2\pi, \ \alpha = 0,1,2,3 \}
\end{align}
and
\begin{align}
\rm{Iso}( \psi^+) = \{ (iX)^\alpha e^{i \theta Z} : \ 0 \leq \theta < 2\pi, \ \alpha = 0,1,2,3 \},
\end{align}
both of which are isomorphic to $K_\infty$.

Given the triplet Bell states have isomorphic isotropy subgroups, their group orbits will have the same shape, $\M(\phi^+) \cong \M(\phi^-) \cong \M(\psi^+) \cong SU(2)/K_\infty$. In the previous section we described how normal subgroups simplify our description of the group orbit, and we can use this here to put the description in terms of $SO(3)$ subgroups. There is a 2-to-1 homomorphism from the quotient $G/\{\mathds{1},-\mathds{1}\}$. Elimination of the $\mathbb{Z}_2$ normal subgroup from the group orbit description gives $SU(2)/K_\infty \cong SO(3)/D_\infty$, where $D_\infty$ is the isotropy subgroup of the cylinder.

\subsection{T-States}

From the Bell states, we can begin to map out the possible shapes of group orbit in the tetrahedron of T-states. Consider the convex hull of triplet Bell states to be the base of the tetrahedron, and $\psi^-$ to be the peak. The vertical height above the base of this tetrahedron indicates the proportion of the singlet state in the convex mixture $p_1 \phi^+ + p_2 \phi^- + p_3 \psi^+ + p_4 \psi^-$, where $p_1 + p_2 + p_3 + p_4 = 1$.

Since $\psi^-$ is symmetric under the $SU(2)$ group action,
\begin{align}
&\M( p_1 \phi^+ + p_2 \phi^- + p_3 \psi^+ + p_4 \psi^-) \nonumber \\
&\quad = \M( p_1 \phi^+ + p_2 \phi^- + p_3 \psi^+),
\end{align}
therefore the possible types of group orbit for T-states will all be exhibited in the convex hull of the triplet states. These can be visualized as a triangle of states, as shown in Fig. \ref{fig:horodecki}.

\begin{figure*}[t]
\includegraphics[width=17cm]{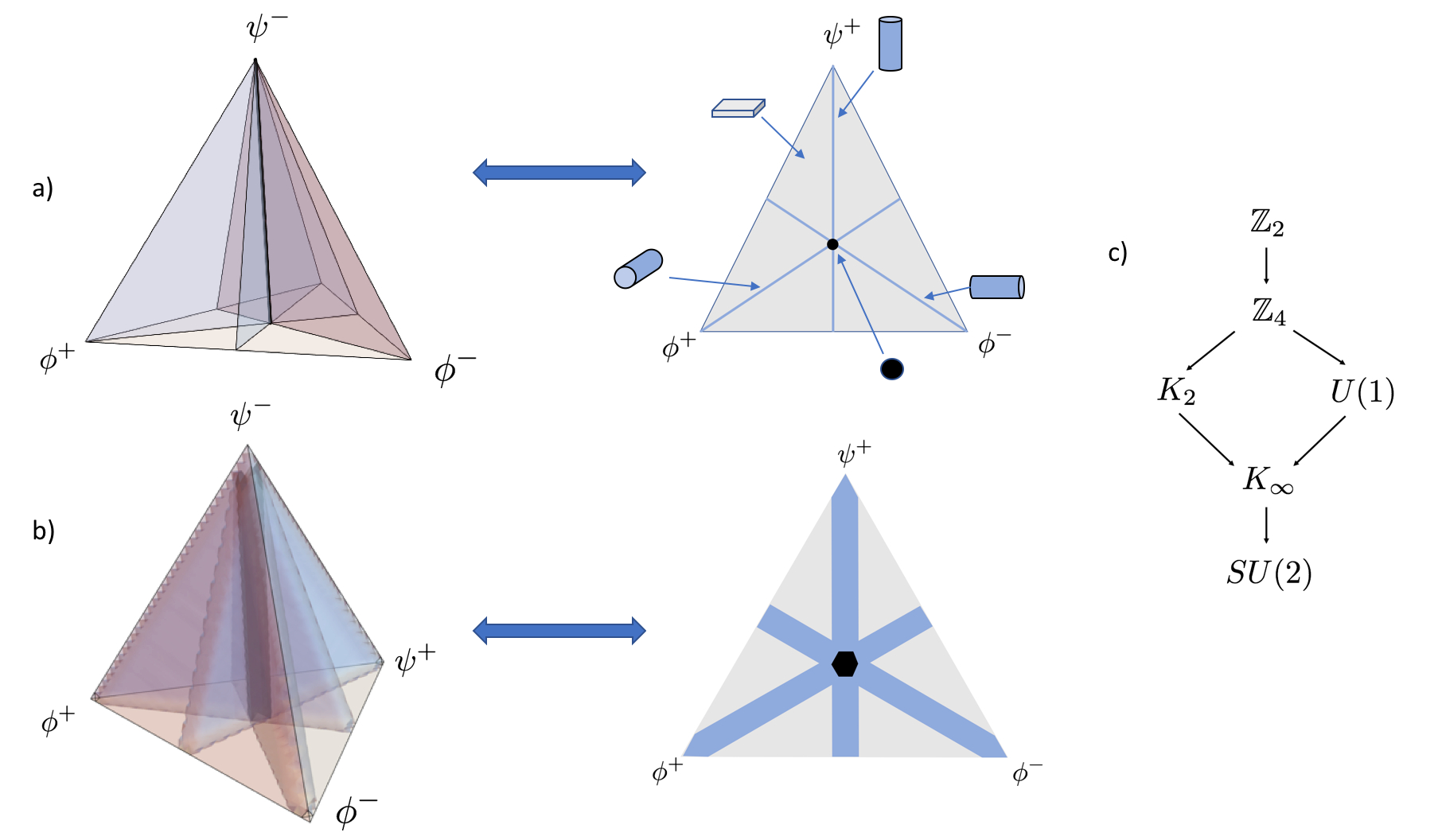}
\caption{\textbf{The partitioning of the $2$-qubit state space.} a) Tetrahedron of T-states coloured to indicate different group orbit shapes: grey -- $SU(2)/K_2 \cong SO(3)/D_2$, light blue -- $SO(3)/D_\infty$, black -- $\{e\}$. Slices through the tetrahedron that are parallel to the plane containing the triplet Bell states exhibit the same structure, exhibited in the triangle to the right. b) Tetrahedron mapping out the smoothed isotropy subgroups with smoothing scale $\epsilon=0.04$. The bands along the diagonals indicate $\rm{Iso}_\epsilon(\rho) \cong K_\infty$, while states in the central hexagon are assigned $\rm{Iso}_\epsilon(\rho) \cong SU(2)$. The remaining states retain their $K_2$ isotropy subgroup. c) Hasse diagram of the observed isotropy subgroups (up to isomorphism) for two qubits. Arrows indicate subset inclusion. This is a sublattice of the full $SU(2)$ subgroup lattice. } \label{fig:horodecki}
\end{figure*}

The vertices of this triangle are the triplet Bell states, and each have a $K_\infty$ isotropy subgroup. Appendix \ref{Bell} also shows that the midpoints of the edges of this triangle, corresponding to states $\frac{1}{4}(\mathds{1} \otimes \mathds{1} + \sigma_i \otimes \sigma_i)$, also have $K_\infty$ isotropy subgroups. Therefore the diagonals of this triangle, corresponding to states 
\begin{align} \label{diagonals}
\rho &= p \phi^+ + (1-p) \frac{1}{4} (\mathds{1} \otimes \mathds{1} + Y \otimes Y) \\
\rho &= p \phi^- + (1-p) \frac{1}{4} (\mathds{1} \otimes \mathds{1} + X \otimes X) \\
\rho &= p \psi^+ + (1-p) \frac{1}{4} (\mathds{1} \otimes \mathds{1} + Z \otimes Z)
\end{align}
for $0 \leq p < 1$, have $K_\infty \prec \rm{Iso}(\rho)$. The states not on these diagonals remain invariant under the subgroup $K_2 = \{ \pm \mathds{1}, \pm i X, \pm i Y, \pm i Z \}$.

The possible isotropy subgroups exhibited by the T-states can be seen with the projectors $\P_H$, as detailed in Appendix \ref{2qubits}. These are:
\begin{itemize}
\item $\rm{Iso}( \rho_T) = SU(2) \Rightarrow \M(\rho) \cong \{e\}$

Werner states along the central axis of the T-state tetrahedron, where $\tau_1 = \tau_2 = \tau_3$. These are the states
\begin{equation}
\rho = (1-p) \ \psi^- + \frac{p}{3} \ ( \psi^+ + \phi^+ + \phi^- )
\end{equation}
for $0 \leq p \leq 1$.

\item $\rm{Iso}( \rho_T) = K_\infty \Rightarrow \M(\rho) \cong SO(3)/D_\infty$

T-states with $\tau_i = \tau_j \neq \tau_k$. These states lie on or above the diagonals of the triangle of triplet Bell states, as seen in Fig.\ref{fig:horodecki}, excluding the Werner states on the central axis of the tetrahedron.

\item $\rm{Iso}( \rho_T) = K_2 \Rightarrow \M(\rho) \cong SO(3)/D_2$

T-states with $\tau_1 \neq \tau_2 \neq \tau_3$. These are states that do not lie on or above the diagonals of the triangle formed by the triplet Bell states.
\end{itemize}

\noindent The smoothed isotropy subgroups are illustrated in the lower part of Fig. \ref{fig:horodecki}, with a smoothing scale of $\epsilon = 0.04$ under the trace distance. Under this smoothing, the horizontal slices through the tetrahedron are no longer equivalent; near  $\psi^-$ all states are assigned $SU(2)$ as their isotropy subgroup, while at the base there remain many states with $\rm{Iso}_\epsilon(\rho) =  K_2$.

\subsection{Maximally Mixed Marginals}

The set of 2-qubit states with maximally mixed marginals is a wider class of states reached from the T-states through local $SU(2)$ unitaries $U(g_1) \otimes U(g_2)$. Appendix \ref{2qubits} details how the isotropy subgroups are enumerated.

If these local unitaries are the same for each qubit, we reach states of the form
\begin{align}
\tilde{\rho}_T &= U(g) \otimes U(g) \ \rho_T \ U^\dagger (g) \otimes U^\dagger (g) = \U_g (\rho_T) \nonumber \\
&= \frac{1}{4}\left(\mathds{1} \otimes \mathds{1} + \sum_{i=1}^3 \tau_i \ \mathbf{c}_i \cdot \sigma \otimes \mathbf{c}_i \cdot \sigma \right),
\end{align}
where the vectors $\{ \mathbf{c}_i \}$ form an orthonormal basis of $\mathbb{R}^3$.

These states behave similarly to the T-states, but with isomorphic isotropy subgroups. This can be seen from the final part of Lemma \ref{concatenate}, where $\text{Iso}(\U_g \circ \rho_T \circ \U_{g^{-1}}) = g [\text{Iso}(\rho_T)] g^{-1}$. Explicitly, instead of $\{ \pm \mathds{1}, \pm i X, \pm i Y , \pm i Z \}$, these rotated T-states will be invariant under $\{\pm \mathds{1}, \ \pm i ( \mathbf{c}_1 \cdot \sigma), \ \pm i ( \mathbf{c}_2 \cdot \sigma), \ \pm i ( \mathbf{c}_3 \cdot \sigma) \} \cong K_2$.

The possible types of group orbit for these states are the same as for the T-states:
\begin{itemize}
\item $\rm{Iso}( \tilde{\rho}_T) = SU(2) \Rightarrow \M(\tilde{\rho}_T) \cong \{e\}$

The same symmetric states as the T-states, unchanged by the unitary transformation $U(g) \otimes U(g)$. These are states $\tilde{\rho}_T$ with $\tau_1 = \tau_2 = \tau_3$.

\item $\rm{Iso}( \tilde{\rho}_T) = K_\infty \Rightarrow \M(\tilde{\rho}_T) \cong SO(3)/D_\infty$

States $\tilde{\rho}_T$ with $\tau_i = \tau_j \neq \tau_k$. These are reached via unitary transformations $U(g) \otimes U(g)$ from the T-states with $K_\infty$ isotropy subgroups.

\item $\rm{Iso}( \tilde{\rho}_T) = K_2 \Rightarrow \M(\tilde{\rho}_T) \cong SO(3)/D_2$

States $\tilde{\rho}_T$ with $\tau_1 \neq \tau_2 \neq \tau_3$. These are reached via unitary transformations $U(g) \otimes U(g)$ from the T-states with $K_2$ isotropy subgroups.
\end{itemize}

\noindent Further states with maximally mixed marginals are reached from the T-states by different local unitaries on each qubit. Suppose that these two local unitaries have the same generator $\mathbf{r} \cdot \sigma$, such that $U(g_1) \otimes U(g_2) = e^{i \theta_1 \mathbf{r} \cdot \sigma} \otimes e^{i \theta_2 \mathbf{r} \cdot \sigma}$. This gives states of the form
\begin{align}
&\rho_M = e^{i \theta_1 \mathbf{r} \cdot \sigma} \otimes e^{i \theta_2 \mathbf{r} \cdot \sigma} \ \rho_T \ e^{-i \theta_1 \mathbf{r} \cdot \sigma} \otimes e^{-i \theta_2 \mathbf{r} \cdot \sigma} \nonumber \\
&= \frac{1}{4}\left(\mathds{1} \otimes \mathds{1} + \tau_1 \ \mathbf{r} \cdot \sigma \otimes \mathbf{r} \cdot \sigma + \sum_{i=2}^3 \tau_i \ \mathbf{c}_i \cdot \sigma \otimes \mathbf{d}_i \cdot \sigma \right)
\end{align}
where $\{\mathbf{r}, \mathbf{c}_2, \mathbf{c}_3 \}$ and $\{\mathbf{r}, \mathbf{d}_2, \mathbf{d}_3 \}$ are two sets of orthonormal bases for $\mathbb{R}^3$. If $U(g_1) = U(g_2)$ in the local unitaries, these sets of bases will be the same and we retrieve $\tilde{\rho}_T$. However, in the case that $\theta_1 \neq \theta_2$ in the local unitaries, new possibilities for the isotropy subgroup are introduced:
\begin{itemize}
\item $\rm{Iso}(\rho_M) = U(1) \Rightarrow \M(\rho_M) \cong S^2$

States $\rho_M$ with $\tau_2 = \tau_3$. These are reached from T-states with $K_\infty$ isotropy subgroups via local unitaries $e^{i \theta_1 \mathbf{r} \cdot \sigma} \otimes e^{i \theta_2 \mathbf{r} \cdot \sigma}$ where $\theta_1 \neq \theta_2$.

\item $\rm{Iso}(\rho_M) = \mathbb{Z}_4 \Rightarrow \M(\rho_M) \cong SO(3)/C_2$

States $\rho_M$ with $\tau_2 \neq \tau_3$. These are reached from T-states with $K_2$ isotropy subgroups via local unitaries $e^{i \theta_1 \mathbf{r} \cdot \sigma} \otimes e^{i \theta_2 \mathbf{r} \cdot \sigma}$ where $\theta_1 \neq \theta_2$.
\end{itemize}

\noindent The states $\tilde{\rho}_T$ are only invariant under a $K_2$ subgroup because the correlation terms are put into diagonal form by the same orthonormal bases $\{ \mathbf{c}_i \}$ and $\{ \mathbf{d}_i \}$ on each qubit. For example, the correlation terms of the T-states are put into the canonical form by the standard cartesian orthonormal basis $\{ \hat{\mathbf{x}}, \hat{\mathbf{y}}, \hat{\mathbf{z}} \}$ on each qubit. When the local unitaries are applied, this can break the $K_2$ symmetry, and the local bases become `misaligned'. However, when $\mathbf{c}_1 = \mathbf{d}_1 = \mathbf{r}$ there is still a subgroup $\mathbb{Z}_4 = \{ \pm \mathds{1} , \pm i (\mathbf{r} \cdot \sigma) \}$ that leaves these states invariant.

The remaining 2-qubit states with maximally mixed marginals are reached from the T-states through local $SU(2)$ unitaries that do not share a generator, such that $U(g_1) \otimes U(g_2) = e^{i \theta_1 \mathbf{r}_1 \cdot \sigma} \otimes e^{i \theta_2 \mathbf{r}_2 \cdot \sigma}$ where $|\mathbf{r}_1| = |\mathbf{r}_2| = 1$ and $\mathbf{r}_1 \neq \mathbf{r}_2$. These states can be expressed in the form
\begin{align}
\tilde{\rho}_M &= e^{i \theta_1 \mathbf{r}_1 \cdot \sigma} \otimes e^{i \theta_2 \mathbf{r}_2 \cdot \sigma}  \ \rho_T \ e^{-i \theta_1 \mathbf{r}_1 \cdot \sigma} \otimes e^{-i \theta_2 \mathbf{r}_2 \cdot \sigma} \nonumber \\
&= \frac{1}{4}\left(\mathds{1} \otimes \mathds{1} + \sum_{i=1}^3 \tau_i \ \mathbf{c}_i \cdot \sigma \otimes \mathbf{d}_i \cdot \sigma \right),
\end{align}
where $\{ \mathbf{c}_i \}$ and $\{ \mathbf{d}_i \}$ are orthonormal bases of $\mathbb{R}^3$ that do not share any elements, i.e. $\mathbf{c}_i \neq \mathbf{d}_j$ for all $i$ and $j$. These states exhibit another type of group orbit:
\begin{itemize}
\item $\rm{Iso}(\rho) = \mathbb{Z}_2 \Rightarrow \M(\rho) \cong SO(3)$

States $\tilde{\rho}_M$. The orthonormal bases $\{ \mathbf{c}_i \}$ and $\{ \mathbf{d}_i \}$ do not share any elements, therefore the only subgroup that leaves these states invariant is $\mathbb{Z}_2 = \{ \pm \mathds{1} \}$. These are in a sense maximally asymmetric states, as they have no non-trivial residual symmetries.
\end{itemize}

\subsection{Local Bloch Vectors}

Finally, the introduction of local Bloch vectors completes the set of 2-qubit states. In the states with maximally mixed marginals, the alignment of the orthonormal bases describing the correlation terms played a key role in determining the symmetry properties. The effect of the local Bloch vectors depends on the how they relate to each other, as well as how they relate to the bases for the correlation terms.

There are three possible isotropy subgroups for these states:
\begin{itemize}
\item $\rm{Iso}(\rho) = \mathbb{Z}_2 \Rightarrow \M(\rho) \cong SO(3)$

This will occur for any states 
\begin{align}
\rho &= \frac{1}{4}\bigg(\mathds{1} \otimes \mathds{1} + \mathbf{a} \cdot \sigma \otimes \mathds{1} + \mathds{1} \otimes \mathbf{b} \cdot \sigma  \nonumber \\
&\qquad \left.  + \sum_{i=1}^3 \tau_i \ \mathbf{c}_i \cdot \sigma \otimes \mathbf{d}_i \cdot \sigma \right)
\end{align}
whose local Bloch vectors $\mathbf{a} \cdot \sigma$ and $\mathbf{b} \cdot \sigma$ are not parallel or anti-parallel, such that $\mathbf{a} \neq \gamma \mathbf{b}$ for any $\gamma \in \mathbb{R}$.

This isotropy subgroup will also occur in states where the basis in which the correlation terms are diagonalised does not align with the local Bloch vectors. These are states of the form
\begin{align}
\rho &= \frac{1}{4}\bigg(\mathds{1} \otimes \mathds{1} + \mathbf{a} \cdot \sigma \otimes \mathds{1} + \mathds{1} \otimes \mathbf{b} \cdot \sigma  \nonumber \\
&\qquad \left. + \tau_1 \ \mathbf{r} \cdot \sigma \otimes \mathbf{r} \cdot \sigma + \sum_{i=2}^3 \tau_i \ \mathbf{c}_i \cdot \sigma \otimes \mathbf{d}_i \cdot \sigma \right)
\end{align}
where $\{\mathbf{r}, \mathbf{c}_2, \mathbf{c}_3 \}$ and $\{\mathbf{r}, \mathbf{d}_2, \mathbf{d}_3 \}$ are two sets of orthonormal bases for $\mathbb{R}^3$ and either $\mathbf{a}$ or $\mathbf{b} \neq \gamma \mathbf{r}$ for any $\gamma \in \mathbb{R}$. In this case, the correlation terms would be invariant under $\mathbb{Z}_4 = \{ \pm \mathds{1}, \pm i \mathbf{r} \cdot \sigma \}$, however the local Bloch vectors prevent this. Similarly if $\tau_2 = \tau_3$ in states of this form, then the local Bloch vectors break a $U(1)$ symmetry, while if $\{\mathbf{c}_i \} = \{ \mathbf{d}_i \}$ then the local Bloch vectors break a $K_2$ symmetry. Likewise if both $\{\mathbf{c}_i \} = \{ \mathbf{d}_i \}$ and $\tau_2 = \tau_3$, the correlation terms are invariant under a $K_\infty$ subgroup that is again broken by the local Bloch vectors.

\item $\rm{Iso}(\rho) = \mathbb{Z}_4 \Rightarrow \M(\rho) \cong SO(3)/C_2$

States of the form
\begin{align}
\rho &= \frac{1}{4}\bigg(\mathds{1} \otimes \mathds{1} + a \mathbf{r} \cdot \sigma \otimes \mathds{1} + b \mathds{1} \otimes \mathbf{r} \cdot \sigma  \nonumber \\
&\qquad \left. + \tau_1 \ \mathbf{r} \cdot \sigma \otimes \mathbf{r} \cdot \sigma + \sum_{i=2}^3 \tau_i \ \mathbf{c}_i \cdot \sigma \otimes \mathbf{d}_i \cdot \sigma \right)
\end{align}
where $\{\mathbf{r}, \mathbf{c}_2, \mathbf{c}_3 \}$ and $\{\mathbf{r}, \mathbf{d}_2, \mathbf{d}_3 \}$ are two sets of orthonormal bases for $\mathbb{R}^3$, $a,b \in \mathbb{R}$ and $\tau_2 \neq \tau_3$. The local Bloch vectors are invariant under the same $\mathbb{Z}_4$ subgroup that leaves the correlation terms invariant. This also includes the case where $\mathbf{c}_i = \mathbf{d}_i$ for $i = 2,3$. The correlation terms are invariant under a $K_2$ subgroup, but the local Bloch vectors break this to a $\mathbb{Z}_4$ subgroup.

Similarly, states of the form
\begin{align}
\rho &= \frac{1}{4}\bigg(\mathds{1} \otimes \mathds{1} + a \mathbf{c}_1 \cdot \sigma \otimes \mathds{1} + b \mathds{1} \otimes \mathbf{c}_1 \cdot \sigma  \nonumber \\
&\qquad \left.  +  \sum_{i=1}^3 \tau_i \ \mathbf{c}_i \cdot \sigma \otimes \mathbf{c}_i \cdot \sigma \right)
\end{align}
where $\{\mathbf{c}_1, \mathbf{c}_2, \mathbf{c}_3 \}$ is an orthonormal basis for $\mathbb{R}^3$ and $\tau_1 = \tau_2 \neq \tau_3$. The correlation terms are invariant under a $K_\infty$ subgroup, but the local Bloch vectors break the residual symmetry down to a $\mathbb{Z}_4$ subgroup as the local Bloch vectors are orthogonal to $\mathbf{c}_3 \cdot \sigma$, the generator of the relevant $U(1)$ subgroup.

\item $\rm{Iso}(\rho) = U(1) \Rightarrow \M(\rho) \cong S^2$

These are states
\begin{align}
\rho &= \frac{1}{4}\bigg(\mathds{1} \otimes \mathds{1} + a \mathbf{r} \cdot \sigma \otimes \mathds{1} + b \mathds{1} \otimes \mathbf{r} \cdot \sigma  \nonumber \\
&\qquad \left. \tau_1 \mathbf{r} \cdot \sigma \otimes \mathbf{r} \cdot \sigma +  \tau \sum_{i=2}^3  \ \mathbf{c}_i \cdot \sigma \otimes \mathbf{d}_i \cdot \sigma \right)
\end{align}
where $\{\mathbf{r}, \mathbf{c}_2, \mathbf{c}_3 \}$ and $\{\mathbf{r}, \mathbf{d}_2, \mathbf{d}_3 \}$ are two sets of orthonormal bases for $\mathbb{R}^3$, $a,b \in \mathbb{R}$. The local Bloch vectors are invariant under the same $U(1)$ subgroup that leaves the correlation terms invariant. This also includes the cases when $\{ \mathbf{c}_i \} = \{\mathbf{d}_i \}$ and when $\tau_1 = \tau$, where the correlation terms are invariant under $K_\infty$ and $SU(2)$ respectively.
\end{itemize}

\noindent This completes the classification of the isotropy subgroups of 2-qubit states.

\subsection{Some simple consequences of the classification}

This classification allows us to infer constraints on the resources required to simulate asymmetric channels under symmetry constraints. If we wish to simulate an axial channel (i.e. $\M(\E) = S^2$) under $SU(2)$ symmetry constraints using the T-states, the resource state must have $\M(\sigma) \succ S^2$. T-states on the diagonals of the tetrahedron, with $\M(\rho) \prec SO(3)/D_\infty$, will not be able to perform such simulations; only the T-states with $\tau_1 \neq \tau_2 \neq \tau_3$ will achieve this task up to some approximation as discussed earlier.

However, convex combinations of different T-states can create more states that break the symmetry to a larger degree. For example, $\rm{Iso}( p \phi^+ + (1-p) \phi^-) = K_2$ when $p \neq 1/2$, despite $\rm{Iso}(\phi^+) \cong \rm{Iso}(\phi^-) \cong K_\infty$.

Group orbits also constrain the dynamics of the T-states. For instance, Lemma \ref{isotropy nondecreasing} implies that symmetric operations move states situated on a diagonal of the tetrahedron only along that diagonal. The T-states are closed under symmetric dynamics -- leaving would entail further symmetries being broken.

Also note that $H_1 \subseteq H_2$ does not imply there exists a symmetric channel from $C(H_1)$ into $C(H_2)$. T-states on the diagonals of the base of the tetrahedron cannot be reached from those that lie off these diagonals under symmetric operations \cite{cristina}, despite $K_2 \prec K_\infty$. Group orbits impose the highest level constraints; further techniques such as those of \cite{cristina} then impose more detailed constraints.

\section{Outlook}

Our analysis has shown that the basic problem of state interconversion under a symmetry constraint has a rich and non-trivial structure. We have established a high-level description of the problem that is consistent with the resource-theoretic picture, however the question whether the framework can find practical use in the same way that superselection rules simplify computation remains to be explored.

We have classified the set of all two-qubit states under the $SU(2)$ tensor product representation, which may be of use in its own right for the study of quantum correlations. However, beyond two-qubits, the general problem can be extremely difficult -- for example even the finite subgroups of $SU(5)$ remain unclassified \cite{stabilizer}. Despite this complexity, we believe that the tools we have developed here provide useful insight into the structure of information processing under symmetry constraints. 

Other techniques for studying the effect of symmetry constraints on quantum operations connect with the present analysis. As already mentioned, recent work on the harmonic analysis of quantum channels suggests that the geometry of group orbits provides an insight into the use of finite quantum resources \cite{cristina}. The present paper has focused purely on the `shape' of this orbit, but the role of the induced geometry on the orbits is only partially understood. To extend such an analysis would require adapting techniques in metrology \cite{metrology} and differential geometry \cite{infogeometry}. It would also be of interest to draw connections with recent work \cite{coarsegraining} that extends classical coarse-graining into the quantum regime with symmetries present.

Finally, these ideas may inform us about other resource theories, using the symmetry constraints directly (e.g. in thermodynamics). Along similar lines, the techniques we have described may find application in the context of realising universal quantum computation via the combination of resource states with simple gate-sets, for example Clifford operations and magic states.

\section{Acknowledgements}
We would like to thank Cristina Cirstoiu, Erick Hinds Mingo, Zoe Holmes and Markus Frembs for many useful discussions. TH is funded by the EPSRC Centre for Doctoral Training in Controlled Quantum Dynamics. DJ is supported by the Royal Society.

\bibliographystyle{unsrtnat}
\bibliography{references}

\onecolumn\newpage
\appendix
\setcounter{theorem}{0}
\setcounter{lemma}{0}

\section{Proofs}\label{proofs}

\begin{lemma}
Given any two quantum channels $\E:\B(\H_A) \rightarrow \B(\H_{A'})$ and $\F : \B(\H_B) \rightarrow \B(\H_{B'})$, with unitary representations of a group $G$ defined on all input and output spaces. Then we have the following
\begin{enumerate}
\item $\rm{Iso}( \E\otimes \F) \succ \rm{Iso}(\E) \wedge \rm{Iso}(\F)$ under the tensor product group action $\mathfrak{U}_g(\E \otimes \F) = \mathfrak{U}_g(\E) \otimes \mathfrak{U}_g(\F)$.
\item If $B' = A$ then $\rm{Iso}(\E \circ \F)  \succ \rm{Iso}(\E) \wedge \rm{Iso}(\F)$.
\item If $A=B$ and $A'=B'$, and $p$ some some probability $0\le p \le 1$, then $\rm{Iso}( p\E + (1-p) \F) \succ \rm{Iso}(\E) \wedge \rm{Iso}(\F)$.
\item $\rm{Iso}(\U_g \circ \E \circ \U_g^\dagger) = g[ \rm{Iso}(\E)]g^{-1}$. 
\end{enumerate}
\end{lemma}

\begin{proof}
For $g \in \rm{Iso}(\E) \wedge \rm{Iso}(\F)$, the global group action gives $\mathfrak{U}_g(\E \otimes \F) = \mathfrak{U}_g(\E) \otimes \mathfrak{U}_g(\F) = \E \otimes \F$, therefore $g \in \rm{Iso}( \E\otimes \F)$. However $\E \otimes \F$ may have further symmetries from taking the tensor product  (for example $\mathrm{Iso}(\ket{0}\bra{0} \otimes \ket{0}\bra{0}) \succ \mathrm{Iso}(\ket{0}\bra{0})$), therefore $\rm{Iso}( \E\otimes \F) \succ \rm{Iso}(\E) \wedge \rm{Iso}(\F)$.

The group action on $\E \circ \F$ is given by $\mathfrak{U}_g(\E \circ \F) = \U_g \circ \E \circ \F \circ \U_{g^{-1}} = \U_g \circ \E \circ \U_{g^{-1}}\circ \U_g \circ \F \circ \U_{g^{-1}} = \mathfrak{U}_g(\E) \circ \mathfrak{U}_g(\F)$. Therefore if $g \in \rm{Iso}(\E) \wedge \rm{Iso}(\F)$, then $g \in \rm{Iso}( \E \circ \F)$, implying $\rm{Iso}(\E \circ \F) \succ \rm{Iso}(\E) \wedge \rm{Iso}(\F)$.

If $g \in \rm{Iso}(\E) \wedge \rm{Iso}(\F)$, then $\mathfrak{U}_g(p \E + (1-p) \F) = p \ \mathfrak{U}_g(\E) + (1-p) \ \mathfrak{U}_g(\F) = p \E + (1-p) \F$, implying $\rm{Iso}(\E) \wedge \rm{Iso}(\F) \prec \rm{Iso}(p \E + (1-p) \F)$.

For $h \in \text{Iso}(\U_g \circ \E \circ \U_{g^{-1}})$, we have $\U_h \circ \U_g \circ \E \circ \U_{g^{-1}} \circ \U_{h^{-1}} = \U_g \circ \E \circ \U_{g^{-1}}$ which simplifies to $\U_{(g^{-1} h g)} \circ \E \circ \U_{(g^{-1} h g)^{-1}}$. Therefore $g^{-1} h g \in \text{Iso}(\E)$ or $h \in g [ \text{Iso}(\E) ] g^{-1}$. When $g h' g^{-1} \in g[\text{Iso}(\E)]g^{-1}$ acts upon $\U_g \circ \E \circ \U_{g^{-1}}$, we have $(\U_g \circ \U_{h'} \circ \U_{g^{-1}}) \circ (\U_g \circ \E \circ \U_{g^{-1}}) \circ (\U_g \circ \U_{h'^{-1}} \circ \U_{g^{-1}}) = \U_g \circ \U_{h'} \circ \E \circ \U_{h'^{-1}} \circ \U_{g^{-1}} = \U_g \circ \E \circ \U_{g^{-1}}$, therefore $g h' g^{-1} \in \text{Iso}(\U_g \circ \E \circ \U_{g^{-1}})$. Therefore $\text{Iso}(\U_g \circ \E \circ \U_{g^{-1}}) = g[\text{Iso}(\E)]g^{-1}$.
\end{proof}

\begin{theorem} \label{simulate}
If a system $B$ in a state $\sigma_B$ can be used to simulate a CPTP map $\E$ under symmetric dynamics, then $\text{Iso}(\sigma_B) \prec \text{Iso}(\E)$. Conversely, if $\E$ has isotropy group $Iso(\E)$ then there exists a quantum system $B$ and quantum state $\sigma_B$ that can be used to simulate $\E$ with $Iso(\sigma_B) = Iso(\E)$.
\end{theorem}

\begin{proof}
Viewing the state $\sigma_B$ as a CPTP map, we can view the simulation as a composite channel $\E = \F \circ \sigma_B$. Therefore $\rm{Iso}(\E) = \rm{Iso}(\F \circ \sigma_B) \succ \rm{Iso}(\F) \wedge \rm{Iso}(\sigma_B) = \rm{Iso}(\sigma_B)$ since $\F$ is symmetric. Therefore $\rm{Iso}(\E) \succ \rm{Iso}(\sigma_B)$, so $\M(\sigma_B) \succ \M(\E)$.

Alternatively, from the Covariant Stinespring theorem \cite{covariantstinespring, marvianthesis}, we can write $\E(\rho) = \F(\rho \otimes \sigma_B) = \rm{tr}_{BC}[V (\rho \otimes \sigma_B \otimes \gamma_C) V^\dagger]$, where $\gamma_C = |\eta\>\<\eta|$ is some pure state, $[V,U_A(g) \otimes U_B(g) \otimes U_C(g)]=0$ and $\rm{Iso}(\gamma_C) =G$. Supposing $g \in \rm{Iso}(\sigma_B)$,
\begin{equation}
\mathcal{E} (\rho) = \mathrm{tr}_{BC} [V (\rho \otimes \U_g(\sigma_B) \otimes \U_g(\gamma_C)) V^\dagger] = \mathrm{tr}_{BC} [V (\mathds{1} \otimes \U_g \otimes \U_g) (\rho \otimes \sigma_B \otimes \gamma_C) V^\dagger].
\end{equation}
$V$ is a symmetric unitary, so $V(\mathds{1} \otimes \U_g \otimes \U_g) = V (\U_g \otimes \U_g \otimes \U_g) (\U_g^\dagger \otimes \mathds{1} \otimes \mathds{1}) 
= (\U_g \otimes \U_g \otimes \U_g) V (\U_g^\dagger \otimes \mathds{1} \otimes \mathds{1})$, giving 
\begin{align}
\mathcal{E}(\rho) &= \mathrm{tr}_{BC} [\mathcal{U}_g ( V ( \mathcal{U}_{g^{-1}} (\rho) \otimes \sigma_B \otimes \gamma_C) V^\dagger)] = \mathcal{U}_g (\mathrm{tr}_{BC} [ V ( \mathcal{U}_{g^{-1}} (\rho) \otimes \sigma_B \otimes \gamma_C) V^\dagger] ) \nonumber \\
&= \mathcal{U}_g \circ \mathcal{E} \circ \mathcal{U}_{g^{-1}} (\rho).
\end{align}
Therefore $\mathrm{Iso}(\sigma_B) \prec \mathrm{Iso}(\mathcal{E})$, and $\M(\sigma_B) \succ \M(\E)$.
 
We now establish the second part of the theorem, and construct the protocol that realises $\E$ with a quantum state that has the same isotropy subgroup as $\E$. Firstly, we note that we can decompose any $\E$ into irreducible process modes \cite{cristina} as
\begin{equation}
\E = \sum_{\lambda,k} \alpha_{\lambda^*,k} (\x) \Phi_{\lambda,k}
\end{equation}
where $\lambda$ labels irreps, including multiplicities and $k$ labels the basis vector of the irrep. As shown in \cite{cristina} the coefficients $\alpha_{\lambda, k}(\x)$ are non-normalised harmonic wavefunctions on the space $\M:=G/\text{Iso}(\E)$, which is diffeomorphic to the orbit of $\E$ under the group $G$, and where the point $\x$ corresponds to the reference data needed to simulate $\E$ under the symmetry constraint.

For our reference system $B$, we take its Hilbert space to be $\mathcal{L}^2 (\M, \mathbb{C})$, where $\sigma_B = |\psi\>\<\psi|$ is a reference state with wavefunction $\psi(\z) = \delta (\z - \x)$, understood in the distributional sense. The protocol to realise $\E$ via $\sigma_B$ involves a measurement of the position $\z$ of $B$, which provides the necessary reference frame data to realise $\E$ on the system $A$. Explicitly, we have
\begin{align}
\E &= \sum_{\lambda,k} \alpha_{\lambda^*,k} (\x) \ \Phi_{\lambda,k} \nonumber \\
& = \tr_B  \sum_{\lambda,k} |\x\>_B\<\x| \otimes \alpha_{\lambda^*,k} (\x)  \ \Phi_{\lambda,k} \nonumber\\
& =  \sum_{\lambda,k}\int dz \ \tr [ M(\z) \ \sigma_B] \ \alpha_{\lambda^*,k} (\z) \ \Phi_{\lambda,k}
\end{align}
where $\{dzM(\z) \}$ is the position measurement on $\M$, which is covariant since $G$ acts transitively on $\M$. Finally, the action of $G$ on $\sigma_B$ is given by the action of $G$ on the quotient space $G/\text{Iso}(\E)$ where at each point we have an isotropy group $\text{Iso}(\E)$, and thus $\text{Iso}(\sigma_B) = \text{Iso}(\E)$ as required.
\end{proof}

\begin{lemma}
Under a symmetric operation $\E$, $\rm{Iso}( \E(\rho)) \succ \rm{Iso}(\rho)$.
\end{lemma}

\begin{proof}
We can view the state $\rho$ as a CPTP map $1 \rightarrow \rho$, therefore from Lemma \ref{concatenate}, $\rm{Iso}( \E(\rho)) = \rm{Iso}( \E \circ \rho) \succ \rm{Iso}(\E) \wedge \rm{Iso}(\rho)$. Since $G \wedge H = H$ for any subgroup $H \prec G$, then $\rm{Iso}(\E) \wedge \rm{Iso}(\rho) = \rm{Iso}(\rho)$, therefore $\rm{Iso}( \E(\rho)) \succ \rm{Iso}(\rho)$, and thus the group orbits also obey $\M( \E(\rho)) \prec \M(\rho)$.
\end{proof}

\begin{lemma}
Let $d(\cdot, \cdot)$ be any metric on the space of quantum states. In terms of this metric we define
\begin{equation}
\mathbf{d}(C(H_1) , C(H_2)) := \inf_{\substack{\sigma_1 \in C(H_1) \\ \sigma_2 \in C(H_2)}} d(\sigma_1 , \sigma_2).
\end{equation}
Then $\mathbf{d}(C(H_1),C(H_2)) = 0$ for all $H_1, H_2 \prec G$. 
\end{lemma}

\begin{proof}
For a state $\rho_1 \in C(H_1)$, $\U_g(\epsilon \rho_1 + (1-\epsilon) \ \sigma) = \epsilon \ \U_g(\rho_1) + (1-\epsilon) \ \U_g(\sigma) = \epsilon \  \U_g(\rho_1) + (1-\epsilon) \ \sigma$ when $\sigma \in C(G)$. Hence $\text{Iso}(\epsilon \rho_1 + (1-\epsilon) \sigma) = \text{Iso}(\rho_1)$, and $\epsilon \rho_1 + (1-\epsilon) \sigma \in C(H_1)$. Likewise, for $\rho_2 \in C(H_2)$, the state $\epsilon \rho_2 + (1-\epsilon) \sigma$ (with the same $\sigma \in C(G)$) is also in $C(H_2)$. As $\epsilon \rightarrow 0$, $d(\epsilon \rho_1 + (1-\epsilon) \sigma, \epsilon \rho_2 + (1-\epsilon) \sigma) \rightarrow 0$, therefore $\mathbf{d}(C(H_1),C(H_2)) = 0$ for all $H_1, H_2 \prec G$.
\end{proof}

\begin{lemma}
The map $\mathcal{P}_H$ has the following properties:
\begin{enumerate}
\item $\mathcal{P}_H$ is the (orthogonal) projector onto $\hat{C}(H)$.
\item $\mathcal{P}_H(\rho) = \arg \min_{\sigma \in \hat{C}(H)} S( \rho || \sigma )$, where $S( \rho || \sigma )$ is the relative entropy.
\end{enumerate}
\end{lemma}

\begin{proof}
For states $\rho \in C(W)$ for $W \succ H$, the unitary $U(h)$ for $h \in H$ will leave the state unchanged. Therefore $\mathcal{P}_H(\rho) = \int_H dh \ \U_h (\rho) = \int_H dh \ \rho = \rho$. Moreover, $\P_H$ is a projector because $\P_H \circ \P_H = \int_H dh \ \int_H dh' \ \U_h \circ \U_{h'} = \int_H dh \ \int_H dh' \ \U_{hh'} = \int_H dh'' \ \U_{h''} = \P_H$. Furthermore, $\P_H^\dagger (\rho) = \int_H dh \ \U_{h^{-1}}(\rho) = \int_H dh \ \U_h(\rho) = \P_H(\rho)$, meaning $\P_H$ is an orthogonal projector.

We follow the proof in \cite{frameness}. By Klein's Inequality, we have $S(\mathcal{P}_H(\rho) || \sigma ) \geq 0$, with equality iff $\sigma = \mathcal{P}_H(\rho)$, therefore $\min_{\sigma \in \hat{C}(H)} S(\mathcal{P}_H(\rho) || \sigma) = 0$. From the definition of relative entropy, $S(\mathcal{P}_H(\rho)) = \min_{\sigma \in \hat{C}(H)} [- \rm{Tr}(\rho \ \log \ \sigma)]$ because $\log$ is analytic and $\sigma \in \hat{C}(H)$. Using the idempotency of $\mathcal{P}_H$, $S(\rho || \mathcal{P}_H(\rho)) = -S(\rho) - \rm{Tr}(\rho \ \log \ \mathcal{P}_H (\rho)) = -S(\rho) + S(\mathcal{P}_H(\rho)) = - S(\rho) + \min_{\sigma \in \hat{C}(H)} [- \rm{Tr}(\rho \ \log \ \sigma)] = \min_{\sigma \in \hat{C}(H)} S( \rho || \sigma )$.
\end{proof}

\begin{lemma}
Given a group action for $G$,
$\rm{Iso}(\mathcal{P}_H) = N_G(H)$, where $N_G(H) = \{g \in G \ : \ gHg^{-1} = H \}$ is the normalizer \cite{normalizer} of $H$ in $G$, and therefore if $H$ is a normal subgroup of $G$ ($H \triangleleft G$) then $\P_H$ is a symmetric operation.
\end{lemma}

\begin{proof}
The group action on $\P_H$ is $\mathfrak{U}_g(\P_H) = \int_H dh \ \mathfrak{U}_g(\U_h) = \int_H dh \ \U_g \circ \U_h \circ \U_{g^{-1}} = \int_H dh \ \U_{ghg^{-1}} $. An element $g \in G$ will therefore be a member of $\text{Iso}(\P_H)$ iff $g H g^{-1} = H$. This is the definition of the normalizer, so we have that $\rm{Iso}(\mathcal{P}_H) = N_G(H)$. Averaging over a normal subgroup $H \triangleleft G$ is therefore symmetric, because $N_G(H) = G$ by the definition of a normal subgroup. 
\end{proof}

\section{Bell States and Single Component T-States}\label{Bell}

The extremal points of the T-states are the Bell states, which can be expressed as the vectorisation \cite{vectorisation} of Pauli matrices:
\begin{align}
\ket{\phi^+} &\propto \ket{\rm{vec}(\mathds{1})}  \\
\ket{\phi^-} &\propto \ket{\rm{vec}(Z)}  \\
\ket{\psi^+} &\propto \ket{\rm{vec}(X)} \\
\ket{\psi^-} &\propto \ket{\rm{vec}(Y)}.
\end{align}
Elements of the isotropy subgroup of a Bell state therefore satisfy $\U_g(\ket{\rm{vec}(\sigma_i)}\bra{\rm{vec}(\sigma_i)}) = \ket{\rm{vec}(\sigma_i)}\bra{\rm{vec}(\sigma_i)}$, hence
\begin{equation}
U(g) \otimes U(g) \ket{\rm{vec}(\sigma_i)} = e^{i \phi} \ket{\rm{vec}(\sigma_i)}
\end{equation}
where $\phi \in \mathbb{R}$ is some arbitrary phase. From the identity $A \otimes B \ket{\rm{vec}(M)} = \ket{\rm{vec}(AMB^T)}$ \cite{vectorisation}, the residual symmetries satisfy
\begin{equation}
U(g) \sigma_i U^T(g) = e^{i \phi} \sigma_i.
\end{equation}
We now consider each of the Bell states in turn.

\subsection{$\rm{Iso}(\phi^+) \cong K_\infty$}

$\ket{\phi^+} \propto \ket{\rm{vec}(\mathds{1})}$ is invariant under group transformations satisfying $U(g) U^T(g) = e^{i \phi} \mathds{1}$. $SU(2)$ group elements take the form $U(g) = \mathrm{exp} (i \varphi [a X + b Y + c Z])$ with $a^2 + b^2 + c^2 = 1$ ($a,b,c \in \mathbb{R}$) and $0 \leq \varphi < 2\pi$, therefore
\begin{equation}
e^{i \varphi (aX + bY + cZ)} e^{i \varphi (aX - bY + cZ)} = e^{i \phi} \mathds{1} \implies e^{i \varphi (aX + bY + cZ)} = e^{i \phi} e^{i \varphi (-aX + bY - cZ)}.
\end{equation}
Expansion of the matrix exponentials gives
\begin{equation}
\cos{\varphi} \mathds{1} + i \sin{\varphi} (aX + bY + cZ) = e^{i \phi} [\cos{\varphi} \mathds{1} + i \sin{\varphi} (-aX + bY - cZ)]
\end{equation}
which provides the conditions
\begin{align}
\cos{\varphi} (1-e^{i \phi}) &= 0 \\
ia \sin{\varphi} (1+e^{i \phi}) &= 0 \\
ib \sin{\varphi} (1-e^{i \phi}) &= 0 \\
ic \sin{\varphi} (1+e^{i \phi}) &= 0.
\end{align}
The problem simplifies because $\phi \in \mathbb{R}$ is restricted in the values it may take. Suppose $\cos{\varphi} \neq 0$, then $\cos{\varphi} (1-e^{i \phi}) = 0$ implies $\phi = 0$ (modulo $2 \pi$). Now suppose $\cos{\varphi} = 0$, ensuring that $\sin{\varphi} \neq 0$. Since $a^2 + b^2 + c^2 = 1$, at least one of $a$, $b$ or $c$ is non-zero, therefore either $1+e^{i \phi} = 0$ or $1-e^{i \phi} = 0$ in order for the conditions to hold. Together, these imply that $\phi = 0$ or $\pi$ (modulo $2\pi$).

When $\phi = 0$, the conditions reduce to
\begin{align}
2i a \sin{\varphi} &= 0 \\
2i c \sin{\varphi} &= 0.
\end{align}
One possible solution is $a=c=0$, leaving $\varphi$ unconstrained and therefore $b= \pm 1$ (without loss of generality choose $b=1$, since $b=-1$ gives `rotations' of the opposite handedness). These solutions take the form $U(g) = e^{i \varphi Y}$. The equations are also satisfied when $\sin{\varphi} = 0$, so when $\varphi = 0,\pi$, which give group transformations $U(g) = \pm \mathds{1}$.

In the $\phi = \pi$ case, the conditions become
\begin{align}
2 \cos{\varphi} &= 0 \\
2i b \sin{\varphi} &= 0.
\end{align}
If $b \neq 0$, then $\cos{\varphi} = \sin{\varphi} = 0$, which is not possible. Therefore the only solutions for these conditions is $b = 0$ and $\cos{\varphi} = 0$, which gives $\varphi = \pi/2$ or $3\pi/2$. In this case $a$ and $c$ are only constrained by $a^2 + c^2 = 1$. The $\varphi=\pi/2$ solution corresponds to the $SU(2)$ transformations
\begin{equation}
U(g) = i(aX + cZ) = (iZ)(c \mathds{1} +i a Y)
= (iZ)(\cos{\theta} \mathds{1} + i \sin{\theta} Y) = (iZ) e^{i \theta Y},
\end{equation}
since $a^2 + c^2 = 1$, so we parametrise in terms of $\cos{\theta}$ and $\sin{\theta}$. Similarly $U(g) = (-iZ) e^{i \theta Y}$ for the $\varphi = 3\pi/2$ case. This shows that
\begin{equation}
\rm{Iso}(\phi^+) = \{ (iZ)^\alpha e^{i \theta Y} : \ 0 \leq \theta < 2\pi, \ \alpha = 0,1,2,3 \} \cong K_\infty.
\end{equation}

\subsection{$\rm{Iso}(\phi^-) = K_\infty$}

Analogous arguments hold for $\ket{\phi^-} \propto \ket{\rm{vec}(Z)}$. The condition for membership of $\rm{Iso}(\phi^-)$ is
\begin{equation}
e^{i \varphi (aX + bY + cZ)} Z e^{i \varphi (aX - bY + cZ)} = e^{i \phi} Z
\end{equation}
which simplifies to
\begin{equation}
e^{i \varphi (aX + bY + cZ)} = e^{i \phi} e^{i \varphi (aX - bY - cZ)}.
\end{equation}
This gives the conditions
\begin{align}
\cos{\varphi}(1-e^{i \phi}) &= 0 \\
i a \sin{\varphi} (1-e^{i \phi}) &= 0 \\
i b \sin{\varphi} (1+e^{i \phi}) &= 0 \\
i c \sin{\varphi} (1+e^{i \phi}) &= 0,
\end{align}
solved in the same way as the previous example. The isotropy subgroup is
\begin{equation}
\rm{Iso}( \phi^-) = \{ (iY)^\alpha e^{i \theta X} : \ 0 \leq \theta < 2\pi, \ \alpha = 0,1,2,3 \} \cong K_\infty.
\end{equation}

\subsection{$\rm{Iso}(\psi^+) = K_\infty$}

$\ket{\psi^+} \propto \ket{\rm{vec}(X)}$ gives the condition
\begin{equation}
e^{i \varphi (aX + bY + cZ)} = e^{i \phi} e^{i \varphi (-aX - bY + cZ)}.
\end{equation}
This is satisfied when
\begin{align}
\cos{\varphi}(1-e^{i \phi}) &= 0 \\
i a \sin{\varphi} (1+e^{i \phi}) &= 0 \\
i b \sin{\varphi} (1+e^{i \phi}) &= 0 \\
i c \sin{\varphi} (1-e^{i \phi}) &= 0,
\end{align}
which gives the isotropy subgroup as
\begin{equation}
\rm{Iso}( \psi^+) = \{ (iX)^\alpha e^{i \theta Z} : \ 0 \leq \theta < 2\pi, \ \alpha = 0,1,2,3 \} \cong K_\infty.
\end{equation}

\subsection{$\rm{Iso}(\psi^-) = SU(2)$}

The singlet Bell state $\ket{\psi^-} \propto \ket{\rm{vec}(Y)}$ behaves differently to the other Bell states. The isotropy subgroup condition is
\begin{equation}
e^{i \varphi (aX + bY + cZ)} = e^{i \phi} e^{i \varphi (aX + bY + cZ)},
\end{equation}
which is satisfied for all elements of $SU(2)$, therefore $\rm{Iso}( \psi^-) = SU(2)$.

\subsection{$\rm{Iso}(\frac{1}{4}(\mathds{1} \otimes \mathds{1} + \tau_i \ \sigma_i \otimes \sigma_i) = K_\infty$}

For states, $\rho = \frac{1}{4}(\mathds{1} \otimes \mathds{1} + \tau_i \sigma_i \otimes \sigma_i)$, the $\mathds{1} \otimes \mathds{1}$ term can be neglected because $\rm{Iso}( \mathds{1} \otimes \mathds{1} + \rho') = \rm{Iso}( \rho')$. For $\rho = \frac{1}{4}(\mathds{1} \otimes \mathds{1} + \tau_1 X \otimes X)$, an $SU(2)$ group element belongs to $\rm{Iso}(\rho)$ when
\begin{equation}
\U_g(X \otimes X) = \U_g(X) \otimes \U_g(X) = X \otimes X,
\end{equation}
satisfied when
\begin{equation}
e^{i \varphi (aX + bY + cZ)} X e^{-i \varphi (aX + bY + cZ)} = \pm X,
\end{equation}
or equivalently
\begin{equation}
e^{i \varphi (aX + bY + cZ)} = \pm e^{i \varphi (aX - bY - cZ)}.
\end{equation}
This is the same condition as for $\psi^+$, therefore 
\begin{equation}
\rm{Iso}( X \otimes X) = \rm{Iso}( \psi^+) = \{ (iY)^\alpha e^{i \theta X} : \ 0 \leq \theta < 2\pi, \ \alpha = 0,1,2,3 \} \cong K_\infty.
\end{equation}
Likewise, elements of $\rm{Iso}( Y \otimes Y)$ satisfy
\begin{equation}
e^{i \varphi (aX + bY + cZ)} = \pm e^{i \varphi (-aX + bY - cZ)},
\end{equation}
the same as for $\phi^+$, therefore
\begin{equation}
\rm{Iso}( Y \otimes Y) = \{ (iZ)^\alpha e^{i \theta Y} : \ 0 \leq \theta < 2\pi, \ \alpha = 0,1,2,3 \} \cong K_\infty,
\end{equation}
while the condition for membership of $\rm{Iso}( Z \otimes Z)$ is
\begin{equation}
e^{i \varphi (aX + bY + cZ)} = \pm e^{i \varphi (-aX + bY - cZ)},
\end{equation}
showing that
\begin{equation}
\rm{Iso}( Z \otimes Z) = \{ (iX)^\alpha e^{i \theta Z} : \ 0 \leq \theta < 2\pi, \ \alpha = 0,1,2,3 \} \cong K_\infty.
\end{equation}

\section{2-Qubit Isotropy Subgroups} \label{2qubits}

The isotropy subgroups of 2-qubit states under a tensor product $SU(2)$ group action can be found with projection operators $\P_H$. These do not indicate which subgroups of $SU(2)$ will appear as isotropy subgroups, however the following lemma restricts the possibilities.
\begin{lemma}\label{twoqubitrestriction}
For a 2-qubit state $\rho$ transforming under an $SU(2)$ tensor product representation, any cyclic subgroup of the isotropy subgroup $\rm{Iso}(\rho)$ must be either $\mathbb{Z}_2$, $\mathbb{Z}_4$ or $U(1)$.
\end{lemma}
\begin{proof}
A group element belongs to $\rm{Iso}(\rho)$ when $\U_g(\rho) = \rho$, or in vectorised form \cite{vectorisation},
\begin{equation}
U(g) \otimes U(g) \otimes U^*(g) \otimes U^*(g) \ket{\rm{vec}(\rho)} = \ket{\rm{vec}(\rho)}.
\end{equation}
For $\rho$ to be invariant under the group transformation $\U_g$, it requires $\ket{\rm{vec}(\rho)}$ to be in $\rm{ker}[U(g) \otimes U(g) \otimes U^*(g) \otimes U^*(g) - \mathds{1}]$.

All groups contain cyclic subgroups from repeated application of a single generator, because $g \in G$ implies that $g^n \in G$. For any cyclic subgroup of $\rm{Iso}(\rho)$, there is some basis in which its representation is diagonal, i.e. $U(g) = e^{i \theta} \ket{0}\bra{0} + e^{-i \theta} \ket{1}\bra{1}$. Therefore
\begin{align}
&U(g) \otimes U(g) \otimes U^*(g) \otimes U^*(g) - \mathds{1} \nonumber \\
&\qquad = (e^{2i\theta}-1) (\ket{1}\bra{1} + \ket{2}\bra{2} + \ket{7}\bra{7} + \ket{11}\bra{11}) + (e^{4i\theta}-1) \ket{3}\bra{3} + (e^{-4i\theta}-1) \ket{12}\bra{12} \nonumber \\
&\hspace{2cm} + (e^{-2i\theta}-1) (\ket{4}\bra{4} + \ket{8}\bra{8} + \ket{13}\bra{13} + \ket{14}\bra{14} ),
\end{align}
where we have assumed for simplicity that $U(g)$ is diagonal in the computational basis (with for example $\ket{2} = \ket{0010}$). Vectors in $\rm{span}(\ket{0},\ket{5},\ket{6},\ket{9},\ket{10},\ket{15})$ are invariant under this $U(1)$ subgroup. Vectors in $\rm{span}(\ket{1},\ket{2},\ket{4},\ket{7},\ket{8},\ket{11}, \ket{13},\ket{14})$ are also in $\rm{ker}[U(g) \otimes U(g) \otimes U^*(g) \otimes U^*(g) - \mathds{1}]$ for $\theta = 0,\pi$, and are therefore invariant under a $\mathbb{Z}_2$ subgroup. Finally, the vectors of $\rm{span}(\ket{3},\ket{12})$ are in $\rm{ker}[U(g) \otimes U(g) \otimes U^*(g) \otimes U^*(g) - \mathds{1}]$ when $\theta = 0, \pi/2, \pi, 3\pi/2$, and hence invariant under a $\mathbb{Z}_4$ subgroup.

This holds for any cyclic subgroup of $SU(2)$, therefore any cyclic subgroup of $\rm{Iso}(\rho)$ must be $\mathbb{Z}_2$, $\mathbb{Z}_2$ or $U(1)$.
\end{proof}

This plays a similar role to the crystallographic restriction theorem in crystallography \cite{tinkham}, and reduces the number of subgroups to check. For example, any 2-qubit state invariant under a $\mathbb{Z}_6$ must be invariant under the $U(1)$ subgroup with the same generator. This allows us to enumerate the possible isotropy subgroups for 2-qubit states. It also suggests that an $N$-qubit state will have only cyclic subgroups $\mathbb{Z}_2, \mathbb{Z}_4, \dots \mathbb{Z}_{2N}, U(1)$ within $\rm{Iso}(\rho)$. 

The channels $\P_H$ provide a way to find the sets $C(H)$. Consider two subgroups $H_1$ and $H_2$, with $H_1 \prec H_2$ and there are no subgroups $H$ that could be the isotropy subgroup of some state such that $H_1 \prec H \prec H_2$. If a state $\rho$ satisfies $\P_{H_1} (\rho) = \rho$ but $\P_{H_2} (\rho) \neq \rho$, then $\rm{Iso}(\rho) = H_1$, and $\rho \in C(H_1)$. We now use this technique to identify the isotropy subgroups of all 2-qubit states.

\subsection{$\mathbb{Z}_2 \prec \rm{Iso}(\rho)$}

All 2-qubit states are symmetric under $\mathbb{Z}_2 = \{ \pm \mathds{1} \}$, confirmed by
\begin{equation}
\P_{\mathbb{Z}_2}(\rho) = \frac{1}{2}[\ \U_\mathds{1}(\rho) + \U_{-\mathds{1}}(\rho) \ ] = \rho
\end{equation}
for any 2-qubit state $\rho$.

\subsection{$\mathbb{Z}_4 \prec \rm{Iso}(\rho)$}

From Lemma \ref{twoqubitrestriction}, if a 2-qubit state has $\mathbb{Z}_3 \prec \rm{Iso}(\rho)$, then it also has $U(1) \prec \rm{Iso}(\rho)$. The next subgroup to check is $\mathbb{Z}_4 = \{ \pm \mathds{1}, \pm i \mathbf{r} \cdot \sigma \}$. For illustrative purposes, let us consider the particular $\mathbb{Z}_4$ subgroup $\{ \pm \mathds{1}, \pm i Z \}$. For a general 2-qubit state $\rho$,
\begin{align}
\P_{\mathbb{Z}_4}(\rho) &= \frac{1}{4}\left[ \mathds{1} \otimes \mathds{1} + \P_{\mathbb{Z}_4}(\mathbf{a} \cdot \sigma \otimes \mathds{1}) + \P_{\mathbb{Z}_4}(\mathds{1} \otimes \mathbf{b} \cdot \sigma) + \sum_{i,j=1}^3 T_{ij} \ \P_{\mathbb{Z}_4} (\sigma_i \otimes \sigma_j) \ \right]  \\
&= \frac{1}{4}\left[ \mathds{1} \otimes \mathds{1} +  a_z \ Z \otimes \mathds{1} + b_z \ \mathds{1} \otimes Z + T_{33} \ Z \otimes Z + \sum_{i,j = 1}^2 T_{ij} \ \sigma_i \otimes \sigma_j  \right].
\end{align}

Now consider the more general $\mathbb{Z}_4 = \{\pm \mathds{1}, \pm i \mathbf{r} \cdot \sigma \}$ subgroup. The useful identities
\begin{align}
(\mathbf{r} \cdot \sigma)(\mathbf{v} \cdot \sigma) &= (\mathbf{r} \cdot \mathbf{v}) \mathds{1} + i (\mathbf{r} \times \mathbf{v}) \cdot \sigma \\
(\mathbf{r} \cdot \sigma)(\mathbf{v} \cdot \sigma)(\mathbf{r} \cdot \sigma) &= 2(\mathbf{r} \cdot \mathbf{v})(\mathbf{r} \cdot \sigma) - \mathbf{v} \cdot \sigma
\end{align}
allow us to calculate $\P_{\mathbb{Z}_4}(\rho)$. The local Bloch vectors give
\begin{align}
\P_{\mathbb{Z}_4}(\mathbf{a} \cdot \sigma \otimes \mathds{1}) &= \frac{1}{2} [ \mathbf{a} \cdot \sigma \otimes \mathds{1} + (\mathbf{r} \cdot \sigma)(\mathbf{a} \cdot \sigma)(\mathbf{r} \cdot \sigma) \otimes \mathds{1} ] \\
&= \frac{1}{2} [ \mathbf{a} \cdot \sigma \otimes \mathds{1} + (2(\mathbf{r} \cdot \mathbf{a})(\mathbf{r} \cdot \sigma) - \mathbf{a} \cdot \sigma) \otimes \mathds{1} ] \\
&= (\mathbf{r} \cdot \mathbf{a})(\mathbf{r} \cdot \sigma) \otimes \mathds{1}
\end{align}
and similarly $\P_{\mathbb{Z}_4}(\mathds{1} \otimes \mathbf{b} \cdot \sigma) = (\mathbf{r} \cdot \mathbf{b})(\mathds{1} \otimes \mathbf{r} \cdot \sigma)$. The correlation terms may be written
\begin{equation}
\sum_{i,j=1}^3 T_{ij} \ \sigma_i \otimes \sigma_j = \sum_{i,j=1}^3 T'_{ij} \ \mathbf{c}_i \cdot \sigma \otimes \mathbf{c}_j \cdot \sigma ,
\end{equation}
where $\{\mathbf{c}_i\}$ is an orthonormal basis for $\mathbb{R}^3$ and we choose $\mathbf{c}_1 = \mathbf{r}$. It is clear that $\P_{\mathbb{Z}_4}(\mathbf{c}_1 \cdot \sigma \otimes \mathbf{c}_1 \cdot \sigma) = \P_{\mathbb{Z}_4}(\mathbf{r} \cdot \sigma \otimes \mathbf{r} \cdot \sigma) = \mathbf{r} \cdot \sigma \otimes \mathbf{r} \cdot \sigma$.
The identity
\begin{align}
(\mathbf{r} \cdot \sigma)(\mathbf{u} \cdot \sigma)(\mathbf{r} \cdot \sigma) \otimes (\mathbf{r} \cdot \sigma)(\mathbf{v} \cdot \sigma)(\mathbf{r} \cdot \sigma) &= 4(\mathbf{r} \cdot \mathbf{u})(\mathbf{r} \cdot \mathbf{v}) \mathbf{r} \cdot \sigma \otimes \mathbf{r} \cdot \sigma - 2 (\mathbf{r} \cdot \mathbf{u}) \mathbf{r} \cdot \sigma \otimes \mathbf{v} \cdot \sigma \nonumber \\
&\qquad -2(\mathbf{r} \cdot \mathbf{v}) \mathbf{u} \cdot \sigma \otimes \mathbf{r} \cdot \sigma + \mathbf{u} \cdot \sigma \otimes \mathbf{v} \cdot \sigma.
\end{align}
allows us to show that when $j \neq 1$,
\begin{align}
\P_{\mathbb{Z}_4}(\mathbf{c}_1 \cdot \sigma \otimes \mathbf{c}_j \cdot \sigma) &= \mathbf{c}_1 \cdot \sigma \otimes \mathbf{c}_j \cdot \sigma + 2(\mathbf{r} \cdot \mathbf{c}_1)(\mathbf{r} \cdot \mathbf{c}_j) \mathbf{r} \cdot \sigma \otimes \mathbf{r} \cdot \sigma - (\mathbf{r} \cdot \mathbf{c}_1) \mathbf{r} \cdot \sigma \otimes \mathbf{c}_j \cdot \sigma -(\mathbf{r} \cdot \mathbf{c}_j) \mathbf{c}_1 \cdot \sigma \otimes \mathbf{r} \cdot \sigma \nonumber \\
&= \mathbf{r} \cdot \sigma \otimes \mathbf{c}_j \cdot \sigma - \mathbf{r} \cdot \sigma \otimes \mathbf{c}_j \cdot \sigma = 0
\end{align}
because $\mathbf{r} \cdot \mathbf{c}_j = \mathbf{c}_1 \cdot \mathbf{c}_j = 0$. The same holds for the $\mathbf{c}_j \cdot \sigma \otimes \mathbf{c}_1 \cdot \sigma$ terms with $j \neq 1$, which also vanish when $\P_{\mathbb{Z}_4}$ is applied. Finally, when $i,j \neq 1$,
\begin{align}
\P_{\mathbb{Z}_4}(\mathbf{c}_i \cdot \sigma \otimes \mathbf{c}_j \cdot \sigma) &= \mathbf{c}_i \cdot \sigma \otimes \mathbf{c}_j \cdot \sigma + 2(\mathbf{r} \cdot \mathbf{c}_i)(\mathbf{r} \cdot \mathbf{c}_j) \mathbf{r} \cdot \sigma \otimes \mathbf{r} \cdot \sigma - (\mathbf{r} \cdot \mathbf{c}_i) \mathbf{r} \cdot \sigma \otimes \mathbf{c}_j \cdot \sigma -(\mathbf{r} \cdot \mathbf{c}_j) \mathbf{c}_i \cdot \sigma \otimes \mathbf{r} \cdot \sigma \nonumber \\
&= \mathbf{c}_i \cdot \sigma \otimes \mathbf{c}_j \cdot \sigma.
\end{align}
Therefore
\begin{align}
\P_{\mathbb{Z}_4}(\rho) &= \frac{1}{4} \left[ \mathds{1} \otimes \mathds{1} + (\mathbf{r} \cdot \mathbf{a})(\mathbf{r} \cdot \sigma) \otimes \mathds{1} + (\mathbf{r} \cdot \mathbf{b})(\mathds{1} \otimes \mathbf{r} \cdot \sigma) + T'_{11} \mathbf{r} \cdot \sigma \otimes \mathbf{r} \cdot \sigma + \sum_{i,j=2}^3 T'_{ij} \ \mathbf{c}_i \cdot \sigma \otimes \mathbf{c}_j \cdot \sigma \ \right] \nonumber \\
&= \frac{1}{4} \left[ \mathds{1} \otimes \mathds{1} + (\mathbf{r} \cdot \mathbf{a})(\mathbf{r} \cdot \sigma) \otimes \mathds{1} + (\mathbf{r} \cdot \mathbf{b})(\mathds{1} \otimes \mathbf{r} \cdot \sigma) + \tau_1 \ \mathbf{r} \cdot \sigma \otimes \mathbf{r} \cdot \sigma + \sum_{i=2}^3 \tau_i \ \mathbf{c}_i \cdot \sigma \otimes \mathbf{d}_i \cdot \sigma \ \right],
\end{align} 
where $\{\mathbf{r}, \mathbf{c}_i \}$ and $\{ \mathbf{r}, \mathbf{d}_i \}$ are orthonormal bases of $\mathbb{R}^3$.

\subsection{$\rm{Iso}(\rho) = \mathbb{Z}_2$}

The image of $\P_{\mathbb{Z}_4}$ forms a proper subset of the image of $\P_{\mathbb{Z}_2}$, therefore there exist 2-qubit states with $\rm{Iso}(\rho) = \mathbb{Z}_2$, and these arise in several ways:
\begin{itemize}
\item States of the form
\begin{align}
\rho &= \frac{1}{4} \left[ \mathds{1} \otimes \mathds{1} + \mathbf{a} \cdot \sigma \otimes \mathds{1} + \mathds{1} \otimes \mathbf{b} \cdot \sigma + \sum_{i=1}^3 \tau_i \ \mathbf{c}_i \cdot \sigma \otimes \mathbf{d}_i \cdot \sigma \ \right],
\end{align}
where $\{\mathbf{c}_i \}$ and $\{ \mathbf{d}_i \}$ are orthonormal bases of $\mathbb{R}^3$ and $\mathbf{c}_i \neq \mathbf{d}_i$ for $i=1,2,3$.
\item States of the form
\begin{align}
\rho &= \frac{1}{4} \left[ \mathds{1} \otimes \mathds{1} + \mathbf{a} \cdot \sigma \otimes \mathds{1} + \mathds{1} \otimes \mathbf{b} \cdot \sigma + \tau_1 \ \mathbf{r} \cdot \sigma \otimes \mathbf{r} \cdot \sigma + \sum_{i=2}^3 \tau_i \ \mathbf{c}_i \cdot \sigma \otimes \mathbf{d}_i \cdot \sigma \ \right],
\end{align}
where $\{\mathbf{r}, \mathbf{c}_i \}$ and $\{\mathbf{r}, \mathbf{d}_i \}$ are orthonormal bases of $\mathbb{R}^3$ but $\mathbf{a}$ and/or $\mathbf{b}$ are not parallel or anti-parallel to $\mathbf{r}$, i.e. there are no $\gamma_1,\gamma_2 \in \mathbb{R}$ such that both $\mathbf{a} = \gamma_1 \mathbf{r}$ and $\mathbf{b} = \gamma_2 \mathbf{r}$.
\item States of the form
\begin{align}
\rho &= \frac{1}{4} \left[ \mathds{1} \otimes \mathds{1} + \mathbf{a} \cdot \sigma \otimes \mathds{1} + \mathds{1} \otimes \mathbf{b} \cdot \sigma + \sum_{i,j=1}^3 T_{ij} \ \sigma_i \otimes \sigma_j \ \right],
\end{align}
where $\mathbf{a} \neq \gamma \mathbf{b}$ for any $\gamma \in \mathbb{R}$, so the local Bloch vectors are neither parallel nor anti-parallel to each other.
\end{itemize}

\subsection{$U(1) \prec \rm{Iso}(\rho)$}

The next subgroup to consider is $U(1)$. We apply $\P_{U(1)}$ to states $\P_{\mathbb{Z}_4}(\rho)$, because if $\mathbb{Z}_4 \nprec \rm{Iso}(\rho)$ then $U(1) \nprec \rm{Iso}(\rho)$. For illustrative purposes, consider the $U(1)$ subgroup $\{ e^{i \theta Z} : \ 0 \leq \theta < 2\pi \}$, which contains $\mathbb{Z}_4 = \{ \pm \mathds{1}, \pm i Z \}$. The states invariant under this $U(1)$ subgroup are
\begin{align}
\P_{U(1)}(\rho) &= \P_{U(1)} \circ \P_{\mathbb{Z}_4} (\rho) \\
&= \frac{1}{4}\left[ \mathds{1} \otimes \mathds{1} +  a_z \ Z \otimes \mathds{1} + b_z \ \mathds{1} \otimes Z + T_{33} \ Z \otimes Z  + \sum_{i,j = 1}^2 T_{ij} \ \P_{U(1)} (\sigma_i \otimes \sigma_j) \  \right]
\end{align}
because the terms $Z \otimes \mathds{1}$, $\mathds{1} \otimes Z$ and $Z \otimes Z$ are certainly invariant under the $U(1)$ subgroup generated by $Z$. For the remaining correlation terms,
\begin{align}
\P_{U(1)}(X \otimes X) &= \frac{1}{2\pi} \int_0^{2\pi} d\theta \ e^{i \theta Z} X e^{-i \theta Z} \otimes e^{i \theta Z} X e^{-i \theta Z} \\
&= \frac{1}{2\pi} \int_0^{2\pi} d\theta \ \cos^2{2\theta} \ X \otimes X + \sin^2{2\theta} \ Y \otimes Y - \cos{2\theta}\sin{2\theta} \ (X \otimes Y + Y \otimes X)  \\
&= \frac{1}{2} (X \otimes X + Y \otimes Y) = \P_{U(1)}(Y \otimes Y)
\end{align}
and similarly $\P_{U(1)}(X \otimes Y) = - \P_{U(1)}(Y \otimes X) = \frac{1}{2} (X \otimes Y - Y \otimes X)$. Together these give
\begin{align}
\P_{U(1)}(\rho) &= \frac{1}{4}\bigg[ \mathds{1} \otimes \mathds{1} +  a_z \ Z \otimes \mathds{1} + b_z \ \mathds{1} \otimes Z + T_{33} \ Z \otimes Z + \frac{1}{2}(T_{11} + T_{22})(X \otimes X + Y \otimes Y) \nonumber \\
&\hspace{2cm} \left. + \frac{1}{2}(T_{12} - T_{21})(X \otimes Y - Y \otimes X) \right]
\end{align}

Now consider the more general $U(1) = \{ e^{i \mathbf{r} \cdot \sigma} : \ 0 \leq \theta < 2\pi \}$. This acts on a Bloch vector as
\begin{align}
e^{i \theta \mathbf{r} \cdot \sigma} (\mathbf{u} \cdot \sigma) e^{-i \theta \mathbf{r} \cdot \sigma} &= \cos{2\theta} \ \mathbf{u} \cdot \sigma + \sin{2\theta} \ (\mathbf{u} \times \mathbf{r}) \cdot \sigma + 2 \sin^2{\theta} \ (\mathbf{r} \cdot \mathbf{u}) \ \mathbf{r} \cdot \sigma,
\end{align}
therefore
\begin{align}
\P_{U(1)}(\mathbf{a} \cdot \sigma \otimes \mathds{1}) &= \frac{1}{2\pi} \int_0^{2\pi} d\theta \ e^{i \theta \mathbf{r} \cdot \sigma} (\mathbf{a} \cdot \sigma) e^{-i \theta \mathbf{r} \cdot \sigma} \otimes \mathds{1} \\
&= \frac{1}{2\pi} \int_0^{2\pi} d\theta \ \cos{2\theta} \ \mathbf{a} \cdot \sigma \otimes \mathds{1} + \sin{2\theta} \ (\mathbf{a} \times \mathbf{r}) \cdot \sigma \otimes \mathds{1} + 2 \sin^2{\theta} \ (\mathbf{r} \cdot \mathbf{a}) \ \mathbf{r} \cdot \sigma \otimes \mathds{1} \nonumber \\
&= (\mathbf{r} \cdot \mathbf{a}) \ \mathbf{r} \cdot \sigma \otimes \mathds{1}
\end{align}
and similarly $P_{U(1)}(\mathds{1} \otimes \mathbf{b} \cdot \sigma) = (\mathbf{r} \cdot \mathbf{b}) \ \mathds{1} \otimes \mathbf{r} \cdot \sigma$. The group action on the correlation terms is
\begin{align}
&e^{i \theta \mathbf{r} \cdot \sigma} (\mathbf{u} \cdot \sigma) e^{-i \theta \mathbf{r} \cdot \sigma} \otimes e^{i \theta \mathbf{r} \cdot \sigma} (\mathbf{v} \cdot \sigma) e^{-i \theta \mathbf{r} \cdot \sigma}  \nonumber \\
&= \cos^2{2\theta} \ \mathbf{u} \cdot \sigma \otimes \mathbf{v} \cdot \sigma + 4 \sin^4{\theta} (\mathbf{r} \cdot \mathbf{u})(\mathbf{r} \cdot \mathbf{v}) \mathbf{r} \cdot \sigma \otimes \mathbf{r} \cdot \sigma + 2 \cos{2\theta} \sin{2\theta} \ [\mathbf{u} \cdot \sigma \otimes (\mathbf{v} \times \mathbf{r}) \cdot \sigma \nonumber \\
&\qquad + (\mathbf{u} \times \mathbf{r}) \cdot \sigma \otimes \mathbf{v} \cdot \sigma ] + 2 \cos{2\theta} \sin^2{\theta} [ (\mathbf{r} \cdot \mathbf{v}) \mathbf{u} \cdot \sigma \otimes \mathbf{r} \cdot \sigma + (\mathbf{r} \cdot \mathbf{u}) \mathbf{r} \cdot \sigma \otimes \mathbf{v} \cdot \sigma] \nonumber \\
&\qquad + 2 \sin{2\theta} \sin^2{\theta} \ [ (\mathbf{r} \cdot \mathbf{v})(\mathbf{u} \times \mathbf{r}) \cdot \sigma \otimes \mathbf{r} \cdot \sigma  + (\mathbf{r} \cdot \mathbf{u}) \mathbf{r} \cdot \sigma \otimes (\mathbf{v} \times \mathbf{r}) \cdot \sigma] \nonumber \\
&\qquad + \sin^2{2\theta} (\mathbf{u} \times \mathbf{r}) \cdot \sigma \otimes (\mathbf{v} \times \mathbf{r}) \cdot \sigma.
\end{align}
This gives
\begin{align}
\P_{U(1)}(\mathbf{c}_i \cdot \sigma \otimes \mathbf{c}_j \cdot \sigma) &= \frac{1}{2\pi} \int_0^{2\pi} d\theta \ e^{i \theta \mathbf{r} \cdot \sigma} (\mathbf{c}_i \cdot \sigma) e^{-i \theta \mathbf{r} \cdot \sigma} \otimes e^{i \theta \mathbf{r} \cdot \sigma} (\mathbf{c}_j \cdot \sigma) e^{-i \theta \mathbf{r} \cdot \sigma}  \\
&= \frac{1}{2} [ \ \mathbf{c}_i \cdot \sigma \otimes \mathbf{c}_j \cdot \sigma - (\mathbf{r} \cdot \mathbf{c}_j) \mathbf{c}_i \cdot \sigma \otimes \mathbf{r} \cdot \sigma - (\mathbf{r} \cdot \mathbf{c}_i) \mathbf{r} \cdot \sigma \otimes \mathbf{c}_j \cdot \sigma  \nonumber \\
&\hspace{2cm} + (\mathbf{c}_i \times \mathbf{r}) \cdot \sigma \otimes (\mathbf{c}_j \times \mathbf{r}) \cdot \sigma + 3 (\mathbf{r} \cdot \mathbf{c}_i)(\mathbf{r} \cdot \mathbf{c}_j) \mathbf{r} \cdot \sigma \otimes \mathbf{r} \cdot \sigma \ ]  \\
&= \frac{1}{2} [ \ \mathbf{c}_i \cdot \sigma \otimes \mathbf{c}_j \cdot \sigma - \delta_{1j} \mathbf{c}_i \cdot \sigma \otimes \mathbf{c}_1 \cdot \sigma - \delta_{1i} \mathbf{c}_1 \cdot \sigma \otimes \mathbf{c}_j \cdot \sigma \nonumber \\
&\qquad +  (\mathbf{c}_i \times \mathbf{c}_1) \cdot \sigma \otimes (\mathbf{c}_j \times \mathbf{c}_1) \cdot \sigma +3 \ \delta_{1i} \delta_{1j} \ \mathbf{c}_1 \cdot \sigma \otimes \mathbf{c}_1 \cdot \sigma ],
\end{align}
therefore
\begin{align}
\P_{U(1)}(\rho) &= \frac{1}{4}\bigg[\mathds{1} \otimes \mathds{1} + (\mathbf{r} \cdot \mathbf{a}) \ \mathbf{r} \cdot \sigma \otimes \mathds{1} + (\mathbf{r} \cdot \mathbf{b}) \ \mathds{1} \otimes \mathbf{r} \cdot \sigma +T'_{11} \ \mathbf{r} \cdot \sigma \otimes \mathbf{r} \cdot \sigma \nonumber \\
&\hspace{1cm} + \frac{1}{2}(T'_{22} + T'_{33}) \ [\mathbf{c}_2 \cdot \sigma \otimes \mathbf{c}_2 \cdot \sigma + \mathbf{c}_3 \cdot \sigma \otimes \mathbf{c}_3 \cdot \sigma ] \nonumber \\
&\hspace{1cm} + \frac{1}{2}(T'_{23} - T'_{32}) \ [ \mathbf{c}_2 \cdot \sigma \otimes \mathbf{c}_3 \cdot \sigma - \mathbf{c}_3 \cdot \sigma \otimes \mathbf{c}_2 \cdot \sigma] \bigg] \\
&= \frac{1}{4}\left[\mathds{1} \otimes \mathds{1} + (\mathbf{r} \cdot \mathbf{a}) \ \mathbf{r} \cdot \sigma \otimes \mathds{1} + (\mathbf{r} \cdot \mathbf{b}) \ \mathds{1} \otimes \mathbf{r} \cdot \sigma + \tau_1 \mathbf{r} \cdot \sigma \otimes \mathbf{r} \cdot \sigma  + \tau \sum_{i=2}^3 \ \mathbf{c}_i \cdot \sigma \otimes \mathbf{d}_i \cdot \sigma \right]
\end{align}
where $\{\mathbf{r}, \mathbf{c}_i\}$ and $\{\mathbf{r}, \mathbf{d}_i\}$ are orthonormal bases of $\mathbb{R}^3$.

\subsection{$K_2 \prec \rm{Iso}(\rho)$}

Consider one such subgroup, $K_2 = \{ \pm \mathds{1}, \pm i X , \pm i Y, \pm i Z \}$. For this subgroup,
\begin{equation}
\P_{K_2}(\rho) = \frac{1}{4} \left[ \mathds{1} \otimes \mathds{1} + \P_{K_2}(\mathbf{a} \cdot \sigma \otimes \mathds{1}) + \P_{K_2}(\mathds{1} \otimes \mathbf{b} \cdot \sigma) + \sum_{i,j = 1}^3 T_{ij} \ \P_{K_2}(\sigma_i \otimes \sigma_j) \  \right].
\end{equation}
Local Bloch vectors are eliminated, since $\P_{K_2}(X \otimes \mathds{1}) = \P_{K_2}(Y \otimes \mathds{1}) = \P_{\K_2}(Z \otimes \mathds{1}) = 0$, and likewise for the other local Bloch vector. For the correlation terms,
\begin{align}
\P_{K_2}(X \otimes X) &= \frac{1}{4} [ X \otimes X + XXX \otimes XXX + YXY \otimes YXY + ZXZ \otimes ZXZ ] = X \otimes X,
\end{align}
and similarly $\P_{K_2}(Y \otimes Y) = Y \otimes Y$ and $\P_{K_2}(Z \otimes Z) = Z \otimes Z$. The remaining correlation terms vanish since
\begin{align}
\P_{K_2}(X \otimes Y) &= \frac{1}{4} [ X \otimes Y + XXX \otimes XYX + YXY \otimes YYY + ZXZ \otimes ZYZ ]  \\
&= \frac{1}{4} [ X \otimes Y - X \otimes Y - X \otimes Y + X \otimes Y] = 0
\end{align}
and similarly for other $\sigma_i \otimes \sigma_j$ terms with $i \neq j$. Therefore
\begin{equation}
\P_{K_2}(\rho) = \frac{1}{4} \left[ \mathds{1} \otimes \mathds{1} + \sum_{i=1}^3 T_{ii} \ \sigma_i \otimes \sigma_i \right],
\end{equation}
which are the T-states. This proves that all T-states $\rho_T$ have $K_2 \prec \rm{Iso}(\rho_T)$.

Consider the more general $K_2$ subgroup, $\{\pm \mathds{1}, \pm i \mathbf{r}_1\cdot\sigma, \pm i \mathbf{r}_2 \cdot \sigma, \pm i \mathbf{r}_3 \cdot \sigma \}$ where $\{\mathbf{r}_1,\mathbf{r}_2,\mathbf{r}_3 \}$ forms an orthonormal basis of $\mathbb{R}^3$. We use that same basis for the correlation terms, such that 2-qubit states can be expressed
\begin{equation}
\rho = \frac{1}{4} \left[ \mathds{1} \otimes \mathds{1} + \mathbf{a} \cdot \sigma \otimes \mathds{1} + \mathds{1} \otimes \mathbf{b} \cdot \sigma  + \sum_{i,j=1}^3 T'_{ij} \ \mathbf{r}_i \cdot \sigma \otimes \mathbf{r}_j \cdot \sigma \ \right].
\end{equation}
We first consider the action of $\P_{K_2}$ on the local Bloch vectors, $\mathbf{a} \cdot \sigma = \sum_{i=1}^3 (\mathbf{r}_i \cdot \mathbf{a}) \ \mathbf{r}_i \cdot \sigma$,
\begin{align}
\P_{K_2}(\mathbf{a} \cdot \sigma \otimes \mathds{1}) &= \frac{1}{4} \bigg[ \mathbf{a} \cdot \sigma \otimes \mathds{1} + \sum_{i=1}^3 (\mathbf{r}_i \cdot \sigma)(\mathbf{a} \cdot \sigma)(\mathbf{r}_i \cdot \sigma) \otimes \mathds{1} \bigg] \\
&= \frac{1}{4}\bigg[\mathbf{a} \cdot \sigma \otimes \mathds{1} + \sum_{i=1}^3 (2 (\mathbf{r}_i \cdot \mathbf{a}) \mathbf{r}_i \cdot \sigma - \mathbf{a} \cdot \sigma) \otimes \mathds{1} \bigg]  \\
&= \frac{1}{2} \left[ \ \sum_{i=1}^3 (\mathbf{r}_i \cdot \mathbf{a}) \  \mathbf{r}_i \cdot \sigma \otimes \mathds{1} - \mathbf{a} \cdot \sigma \otimes \mathds{1} \ \right] = 0,
\end{align}
and similarly $\P_{K_2}(\mathds{1} \otimes \mathbf{b} \cdot \sigma) = 0$. The correlation terms give
\begin{align}
\P_{K_2}(\mathbf{r}_i \cdot \sigma \otimes \mathbf{r}_j \cdot \sigma) &= \frac{1}{4} \left[ \mathbf{r}_i \cdot \sigma \otimes \mathbf{r}_j \cdot \sigma + \sum_{k=1}^3 (\mathbf{r}_k \cdot \sigma)(\mathbf{r}_i \cdot \sigma)(\mathbf{r}_k \cdot \sigma) \otimes (\mathbf{r}_k \cdot \sigma)(\mathbf{r}_j \cdot \sigma)(\mathbf{r}_k \cdot \sigma) \right]  \\
&= \frac{1}{4} \left[ \mathbf{r}_i \cdot \sigma \otimes \mathbf{r}_j \cdot \sigma + \sum_{k=1}^3 \bigg(4 (\mathbf{r}_k \cdot \mathbf{r}_i)(\mathbf{r}_k \cdot \mathbf{r}_j) \mathbf{r}_k \cdot \sigma \otimes \mathbf{r}_k \cdot \sigma + \mathbf{r}_i \cdot \sigma \otimes \mathbf{r}_j \cdot \sigma \right. \nonumber \\
&\hspace{5cm} - 2(\mathbf{r}_k \cdot \mathbf{r}_i) \mathbf{r}_k \cdot \sigma \otimes \mathbf{r}_j \cdot \sigma - 2(\mathbf{r}_k \cdot \mathbf{r}_j) \mathbf{r}_i \cdot \sigma \otimes \mathbf{r}_k \cdot \sigma \bigg) \bigg] \nonumber  \\
&= \frac{1}{4} \left[ \sum_{k=1}^3 (4 \delta_{ik}\delta_{jk} \mathbf{r}_k \cdot \sigma \otimes \mathbf{r}_k \cdot \sigma - 2 \delta_{ik} \mathbf{r}_k \cdot \sigma \otimes \mathbf{r}_j \cdot \sigma  - 2 \delta_{jk} \mathbf{r}_i \cdot \sigma \otimes \mathbf{r}_k \cdot \sigma ) + 4 \mathbf{r}_i \cdot \sigma \otimes \mathbf{r}_j \cdot \sigma \right] \nonumber \\
&= \delta_{ij} \ \mathbf{r}_j \cdot \sigma \otimes \mathbf{r}_j \cdot \sigma
\end{align}
Therefore
\begin{equation}
\P_{K_2}(\rho) = \frac{1}{4} \left[ \mathds{1} \otimes \mathds{1} + \sum_{i = 1}^3 \tau_i \ \mathbf{r}_i \cdot \sigma \otimes  \mathbf{r}_i \cdot \sigma \right]
\end{equation}
where $\{\mathbf{r}_i\}$ is an orthonormal basis for $\mathbb{R}^3$.

\subsection{$\rm{Iso}(\rho) = \mathbb{Z}_4$}

We have now calculated $\P_H(\rho)$ for the two subgroups directly above $\mathbb{Z}_4$ on the subgroup lattice that are possible isotropy subgroups. This identifies the 2-qubit states in $C(\mathbb{Z}_4)$. They are:
\begin{itemize}
\item States of the form
\begin{equation}
\rho = \frac{1}{4} \bigg[ \mathds{1} \otimes \mathds{1} + \tau_1 \ \mathbf{r} \cdot \sigma \otimes \mathbf{r} \cdot \sigma + \sum_{i=2}^3 \tau_i \ \mathbf{c}_i \cdot \sigma \otimes \mathbf{d}_i \cdot \sigma \bigg]
\end{equation}
where $\tau_2 \neq \tau_3$ and $\mathbf{c}_i \neq \mathbf{d}_i$ for $i=2,3$.
\item States of the form
\begin{equation}
\rho = \frac{1}{4} \left[ \mathds{1} \otimes \mathds{1} + a \mathbf{r} \cdot \sigma \otimes \mathds{1} + b \mathds{1} \otimes \mathbf{r} \cdot \sigma + \tau_1 \ \mathbf{r} \cdot \sigma \otimes \mathbf{r} \cdot \sigma  + \sum_{i=2}^3 \tau_i \ \mathbf{c}_i \cdot \sigma \otimes \mathbf{d}_i \cdot \sigma \right]
\end{equation}
where $\tau_2 \neq \tau_3$ and $\mathbf{c}_i \neq \mathbf{d}_i$ for $i=2,3$, and $a,b \in \mathbb{R}$.
\item The states
\begin{equation}
\rho = \frac{1}{4}\left[\mathds{1} \otimes \mathds{1} + a \mathbf{c}_1 \cdot \sigma \otimes \mathds{1} + b \mathds{1} \otimes \mathbf{c}_1 \cdot \sigma   +  \sum_{i=1}^3 \tau_i \ \mathbf{c}_i \cdot \sigma \otimes \mathbf{c}_i \cdot \sigma \right]
\end{equation}
where $a,b \in \mathbb{R}$. This gives a $\mathbb{Z}_4$ isotropy subgroup when $\tau_1 = \tau_2 \neq \tau_3$.
\end{itemize}

\subsection{$K_\infty \prec \rm{Iso}(\rho)$}

The next subgroup to consider is $K_\infty$, which contains the $K_2$ subgroup. Therefore we need only apply $\P_{K_\infty}$ to states $\P_{K_2}(\rho)$. We can also say the same about states $\P_{U(1)}(\rho)$, however the $K_2$ subgroup projects onto a simpler subset of states.

As an illustrative example, consider the particular subgroup $\{ (iX)^\alpha e^{i \theta Z} : \ 0 \leq \theta < 2\pi, \ \alpha = 0,1,2,3 \}$. This contains the $K_2$ subgroup $\{\pm \mathds{1}, \pm iX, \pm iY, \pm iZ \}$, so we need only apply $\P_{K_\infty}$ to the T-states. The CPTP map $\P_{K_\infty}$ can be split up as $\P_{K_\infty} = \P_{U(1)} \circ \P_{\mathbb{Z}_4}$, therefore
\begin{equation}
\P_{K_\infty}(\rho_T) = \P_{U(1)}(\rho_T) = \frac{1}{4} \bigg( \mathds{1} \otimes \mathds{1} + \frac{1}{2}(\tau_1 + \tau_2) \ (X \otimes X + Y \otimes Y) + \tau_3 \ Z \otimes Z \bigg)
\end{equation}

\noindent More generally, we consider 
\begin{equation}
K_\infty = \{ (i \ \mathbf{r}_j \cdot \sigma)^\alpha e^{i \theta \ \mathbf{r}_k \cdot \sigma} : \ 0 \leq \theta < 2\pi, \ \alpha = 0,1,2,3 \}
\end{equation}
where $\{\mathbf{r}_i \}$ is an orthonormal basis of $\mathbb{R}^3$ and $j \neq k$. The projector $\P_{K_\infty}$ on a general 2-qubit state gives
\begin{align}
\P_{K_\infty} (\rho) &= \P_{U(1)} \circ \P_{K_2} (\rho) = \P_{U(1)} \left[ \frac{1}{4} \left( \mathds{1} \otimes \mathds{1} + \sum_{i=1}^3 \tau_i \ \mathbf{r}_i \cdot \sigma \otimes \mathbf{r}_i \cdot \sigma \right) \right]  \\
&= \frac{1}{4} \left( \mathds{1} \otimes \mathds{1} + \frac{1}{2} \left[ \sum_{i \neq k} \tau_i \right] \left[ \sum_{i \neq k} \mathbf{r}_i \cdot \sigma \otimes \mathbf{r}_i \cdot \sigma \right] + \tau_k \ \mathbf{r}_k \cdot \sigma \otimes \mathbf{r}_k \cdot \sigma \right)
\end{align}

\subsection{$\rm{Iso}(\rho) = K_2$}

These are the T-states and the states with maximally mixed marginals that can be reached from the T-states through rigid $SU(2)$ rotations:
\begin{equation}
\rho = \frac{1}{4} \bigg[ \mathds{1} \otimes \mathds{1} + \sum_i \tau_i \ \mathbf{c}_i \cdot \sigma \otimes \mathbf{c}_i \cdot \sigma \bigg],
\end{equation}
with $\tau_1 \neq \tau_2 \neq \tau_3$. There cannot be any local Bloch vectors.

\subsection{$\rm{Iso}(\rho) = U(1)$}

These are states of the form
\begin{equation}
\rho = \frac{1}{4} \left[ \mathds{1} \otimes \mathds{1} + \tau_1 \ \mathbf{r} \cdot  \sigma \otimes \mathbf{r} \cdot \sigma + \tau \sum_{i=2}^3 \mathbf{c}_i \cdot \sigma \otimes \mathbf{d}_i \cdot \sigma \right]
\end{equation}
where $\mathbf{c}_i \neq \mathbf{d}_i$ for $i=2,3$. Local Bloch vectors are permitted so long as they are either aligned or anti-aligned with the shared Bloch vector in the correlation terms:
\begin{equation}
\rho = \frac{1}{4} \left[ \mathds{1} \otimes \mathds{1} + a \mathbf{r} \cdot \sigma \otimes \mathds{1} + b \mathds{1} \otimes \mathbf{r} \cdot \sigma + \tau_1 \ \mathbf{r} \cdot  \sigma \otimes \mathbf{r} \cdot \sigma + \tau \sum_{i=2}^3 \mathbf{c}_i \cdot \sigma \otimes \mathbf{d}_i \cdot \sigma \right]
\end{equation}
where $\mathbf{c}_i \neq \mathbf{d}_i$ for $i=2,3$ and $a,b \in \mathbb{R}$.

\subsection{$\rm{Iso}(\rho) = SU(2)$}

Since $G = SU(2)$, when we apply $\P_{SU(2)}$ to a general 2-qubit state we get the form of a general 2-qubit symmetric state. Since $K_\infty \prec SU(2)$,
\begin{equation}
\P_{SU(2)}(\rho) = \P_{SU(2)} \circ \P_{K_\infty} (\rho) = \frac{1}{4} \left( \mathds{1} \otimes \mathds{1} + \sum_{i=1}^3 \tau_i \ \P_{SU(2)} \ [ \mathbf{c}_i \cdot \sigma \otimes \mathbf{c}_i \cdot \sigma ] \right).
\end{equation}

A general $SU(2)$ group transformation is $U(g) = e^{i \varphi \mathbf{r} \cdot \sigma}$. We parameterise $\mathbf{r} = ( \sin{\phi} \cos{\theta}, \sin{\phi} \sin{\theta}, \cos{\phi})$. Considering each of the terms individually,
\begin{align}
&\P_{SU(2)}(\mathbf{c}_i \cdot \sigma \otimes \mathbf{c}_i \cdot \sigma ) \nonumber \\
&= \frac{1}{2\pi^2} \int_0^{2\pi} d\theta \int_0^\pi d\phi \ \sin{\phi} \int_0^\pi d\varphi \ \sin^2{\varphi} \ \bigg( e^{i \varphi \mathbf{r} \cdot \sigma} \mathbf{c}_i \cdot \sigma e^{-i \varphi \mathbf{r} \cdot \sigma} \otimes e^{i \varphi \mathbf{r} \cdot \sigma} \mathbf{c}_i \cdot \sigma e^{-i \varphi \mathbf{r} \cdot \sigma} \bigg) \\
&= \frac{1}{2\pi^2} \int_0^{2\pi} d\theta \int_0^\pi d\phi \ \sin{\phi} \int_0^\pi d\varphi \ \sin^2{\varphi} \cos^2{2\varphi} \ \mathbf{c}_i \cdot \sigma \otimes \mathbf{c}_i \cdot \sigma \nonumber \\
&\hspace{1.5cm} + 2 \sin^2{\varphi} \cos{2\varphi} \sin{2\varphi} \ [\mathbf{c}_i \cdot \sigma \otimes (\mathbf{c}_i \times \mathbf{r}) \cdot \sigma + (\mathbf{c}_i \times \mathbf{r}) \cdot \sigma \otimes \mathbf{c}_i \cdot \sigma ] \nonumber \\
&\hspace{1.5cm} + 2 \cos{2\varphi} \sin^4{\varphi} \ (\mathbf{r} \cdot \mathbf{c}_i) \ [ \mathbf{c}_i \cdot \sigma \otimes \mathbf{r} \cdot \sigma + \mathbf{r} \cdot \sigma \otimes \mathbf{c}_i \cdot \sigma ] \nonumber \\
&\hspace{1.5cm} + 2 \sin{2\varphi} \sin^4{\varphi} \ (\mathbf{r} \cdot \mathbf{c}_i) \ [ (\mathbf{c}_i \times \mathbf{r}) \cdot \sigma \otimes \mathbf{r} \cdot \sigma + \mathbf{r} \cdot \sigma \otimes (\mathbf{c}_i \times \mathbf{r}) \cdot \sigma ] \nonumber \\
&\hspace{1.5cm} + 4 \sin^6{\varphi} \ (\mathbf{r} \cdot \mathbf{c}_i)^2 \ \mathbf{r} \cdot \sigma \otimes \mathbf{r} \cdot \sigma  + \sin^2{\varphi} \sin^2{2\varphi} \ (\mathbf{c}_i \times \mathbf{r}) \cdot \sigma \otimes (\mathbf{c}_i \times \mathbf{r}) \cdot \sigma  \\
&= \frac{1}{8\pi} \int_0^{2\pi} d\theta \int_0^\pi d\phi \ \sin{\phi} \  \bigg( \mathbf{c}_i \cdot \sigma \otimes \mathbf{c}_i \cdot \sigma - 2 (\mathbf{r} \cdot \mathbf{c}_i) \ [ \mathbf{c}_i \cdot \sigma \otimes \mathbf{r} \cdot \sigma + \mathbf{r} \cdot \sigma \otimes \mathbf{c}_i \cdot \sigma ] \nonumber \\
&\hspace{5cm} + (\mathbf{c}_i \times \mathbf{r}) \cdot \sigma \otimes (\mathbf{c}_i \times \mathbf{r}) \cdot \sigma + 5 (\mathbf{r} \cdot \mathbf{c}_i)^2 \ \mathbf{r} \cdot \sigma \otimes \mathbf{r} \cdot \sigma \bigg)
\end{align}
Performing the integrals over $\theta$ and $\phi$ gives
\begin{align}
\P_{SU(2)}(\rho) = \frac{1}{4} \bigg[ \ \mathds{1} \otimes \mathds{1} + \tau \sum_{i=1}^3 \sigma_i \otimes \sigma_i \ \bigg].
\end{align}

\subsection{$\rm{Iso}(\rho) = K_\infty$}
These are states in the image of $\P_{K_\infty}$ but not in that of $\P_{SU(2)}$. They have the form
\begin{equation}
\rho = \frac{1}{4} \left[ \mathds{1} \otimes \mathds{1} + \sum_i \tau_i \ \mathbf{c}_i \cdot \sigma \otimes \mathbf{c}_i \cdot \sigma \right]
\end{equation}
where $\tau_1 = \tau_2 \neq \tau_3$ i.e. two of the $\tau_i$ are the same, but not equal to the third.

\end{document}